\documentclass{amsart}
\usepackage[usenames,dvipsnames]{color}
\usepackage{amsmath,amsthm,amsfonts,amssymb,multicol,amscd,amsbsy}
\usepackage{graphicx}
\usepackage{booktabs}
\usepackage{caption}
\usepackage{array}
\usepackage{enumerate}
\usepackage{url}
\usepackage[nodayofweek]{datetime}
\usepackage[english]{babel}
\usepackage{mathtools}
\usepackage[toc,page]{appendix}
\usepackage{verbatim} 
\usepackage{amsthm} 
\usepackage{enumitem}
\usepackage{bbm} 
\usepackage[normalem]{ulem} 
\usepackage{bm}
\usepackage{hyperref}
\usepackage{tabularx}
\usepackage{latexsym}

\DeclarePairedDelimiter\floor{\lfloor}{\rfloor}

\newcommand{\be}{\begin}
\newcommand{\e}{\end}
\newcommand{\beq}{\begin{equation}}
\newcommand{\eeq}{\end{equation}}

\newcommand{\ul}{\underline}
\renewcommand{\l}{\left}
\renewcommand{\r}{\right}

\renewcommand{\d}{\mathrm{d}} 

\newcommand{\set}[1]{\mathbb{#1}}
\newcommand{\curly}[1]{\mathcal{#1}}

\newcommand{\setof}[2]{\left\{ #1\; : \;#2 \right\}}

\newcommand{\R}{\set{R}}

\newcommand{\Z}{\set{Z}}
\newcommand{\E}{\set{E}}

\newcommand{\T}{\set{T}}

\newcommand{\om}{\omega}

\newcommand{\eps}{\epsilon}

\newcommand{\lam}{\lambda}
\newcommand{\Lam}{\Lambda}

\newcommand{\sig}{\sigma}
\newcommand{\al}{\alpha}
\newcommand{\de}{\delta}

\newcommand{\ind}{\mathbbm{1}}		

\newcommand{\ttmatrix}[4]{\left(\be{array}{cc} #1&#2\\	#3&#4 \e{array}	\right)}

\newcommand{\Tr}{\mathrm{Tr}}	

\newtheorem{thm}{Theorem}[section]
\newtheorem{lm}[thm]{Lemma}

\newtheorem{prop}[thm]{Proposition}

\theoremstyle{definition}
\newtheorem{defn}[thm]{Definition}
\newtheorem*{defn*}{Definition}

\numberwithin{equation}{section}

\theoremstyle{remark}

\def\dotuline{\bgroup
  \ifdim\ULdepth=\maxdimen  
   \settodepth\ULdepth{(j}\advance\ULdepth.4pt\fi
  \markoverwith{\begingroup
  \advance\ULdepth0.08ex
  \lower\ULdepth\hbox{\kern.15em .\kern.1em}%
  \endgroup}\ULon}

\def\dashuline{\bgroup
  \ifdim\ULdepth=\maxdimen  
   \settodepth\ULdepth{(j}\advance\ULdepth.4pt\fi
  \markoverwith{\kern.15em
  \vtop{\kern\ULdepth \hrule width .3em}%
  \kern.15em}\ULon}
\allowdisplaybreaks

\begin{document}
\title[Eigenvalue distribution of matrices defined by skew-shift]{Global eigenvalue distribution of matrices defined by the skew-shift}
\author{Arka Adhikari}
\author{Marius Lemm}
\author{Horng-Tzer Yau}
\address{Department of Mathematics, Harvard University, Cambridge, MA 02138, USA}
%
%
\begin{abstract}
 We consider large Hermitian matrices whose entries are defined by evaluating the exponential function along orbits of the skew-shift $\binom{j}{2} \om+jy+x \mod 1$ for irrational $\om$. We prove that the eigenvalue distribution of these matrices converges to the corresponding distribution from random matrix theory on the global scale, namely, the Wigner semicircle law for square matrices and the Marchenko-Pastur law for rectangular matrices. The results evidence the quasi-random nature of the skew-shift dynamics which was observed in other contexts by Bourgain-Goldstein-Schlag and Rudnick-Sarnak-Zaharescu. 
\end{abstract}
\maketitle
\section{Introduction and main results}
The Wigner semicircle law was the first derived example of random matrix statistics. In 1955, Wigner showed that it arises as the asymptotic density of eigenvalues of $N\times N$ Hermitian random matrices $H_N$, as $N\to\infty$ \cite{Wigner}. In recent years, extensive efforts have been devoted to deriving the Wigner semicircle law down to very small scales for large classes of random matrices, including sparse ones coming from random graphs containing relatively few random variables (essentially $\gg N$); see \cite{BKY,BK,EKYY,ESY} and others.

In the present paper, we study the following question:\\

\emph{Suppose the entries of the large Hermitian matrices $H_N$ are generated by sampling along the orbits of an ergodic dynamical system. Do their eigenvalues  still exhibit random matrix statistics, like the Wigner semicircle law?}\\

We will answer this question in the affirmative for the model of $H_N$ defined below, where the underlying dynamical system is generated from the skew-shift dynamics: 
$$
\binom{j}{2} \om+jy+x \mod 1,
$$
Here $x,y\in \T$ (with $\T$ being the torus) are the starting positions of the dynamical system and $\om\in \T$ is a (typically irrational) parameter called the frequency. The skew-shift dynamics possesses only weak ergodicity properties, e.g., it is not even weakly mixing. Nonetheless, it is believed to behave in a quasi-random way (meaning like an i.i.d.\ sequence of random variables) in various ways reviewed at the end of the introduction. Moreover, the quasi-random behavior of the skew-shift should deviate from that of the more rigid standard shift  $j\om+x\mod 1$ (the circle rotation by an irrational angle $\om$). The key difference between the skew-shift and circle rotation is of course the appearance of a quadratic term $j^2\om$ for the skew-shift. This quadratic term has the effect of increasing the oscillations and thus improving the decay of the exponential sums over skew-shift orbits. This general fact is a central tenet of analytic number theory \cite{HL,Mont,W1,W2}, and of our analysis here as well.

We consider $2N\times 2N$ Hermitian matrices of the form
$$
H=\ttmatrix{0}{X}{X^*}{0}
$$
with $X$ an $N\times N$ matrix generated from the skew-shift via
\beq\label{eq:Xdefn}
X_{i,j}=\frac{1}{\sqrt{N}} e\l[ \l(\binom{j}{2} \om_i+jy_i+x_i\r)\r],\qquad e[t]:=\exp(2\pi \sqrt{-1} t)
\eeq
Here the $\om_1,\ldots,\om_N$ are chosen deterministically (see the examples below), the $y_1,\ldots,y_N$ in \eqref{eq:Xdefn} are sampled uniformly and independently from $[0,1]$, and the $x_1,\ldots,x_N$ are arbitrary. (In particular, one can take $x_1=x_2=\ldots=x_N=0$.)

We prove that the empirical spectral distribution of $H$ converges weakly to the Wigner semicircle law, that is,
$$
\frac{1}{2N}\sum_{j=1}^{2N} \delta_{\lam_j}\,\rightharpoonup\, \frac{1}{2\pi} \sqrt{4-x^2}\d x,\qquad\textnormal{ as } N\to\infty.
$$

The result applies for all choices of frequencies $(\om_1,\om_2,\ldots)$ for which a certain oscillatory exponential sum is of small size; see Definition \ref{defn:qrandom} below. We can verify that various classes of frequencies $(\om_1,\om_2,\ldots)$ with sufficient irrationality properties fall under this definition (see Section \ref{sect:ntheory}). Two examples of viable choices for the frequencies are the irrational circle rotation
$$
(\om_1,\om_2,\om_3,\om_4,\ldots)=(\sqrt{2},\,2\sqrt{2},\,3\sqrt{2},\,4\sqrt{2},\,\ldots)
$$
(see Figure \ref{fig:sqrt2}) and the square-root sequence
$$
(\om_1,\om_2,\om_3,\om_4,\ldots)=\l(1,\sqrt{2},\sqrt{3},\sqrt{4},\ldots\r)
$$
(see Figure \ref{fig:3/2}). 

\begin{figure}[t]
\begin{center}
    \includegraphics[scale=.2]{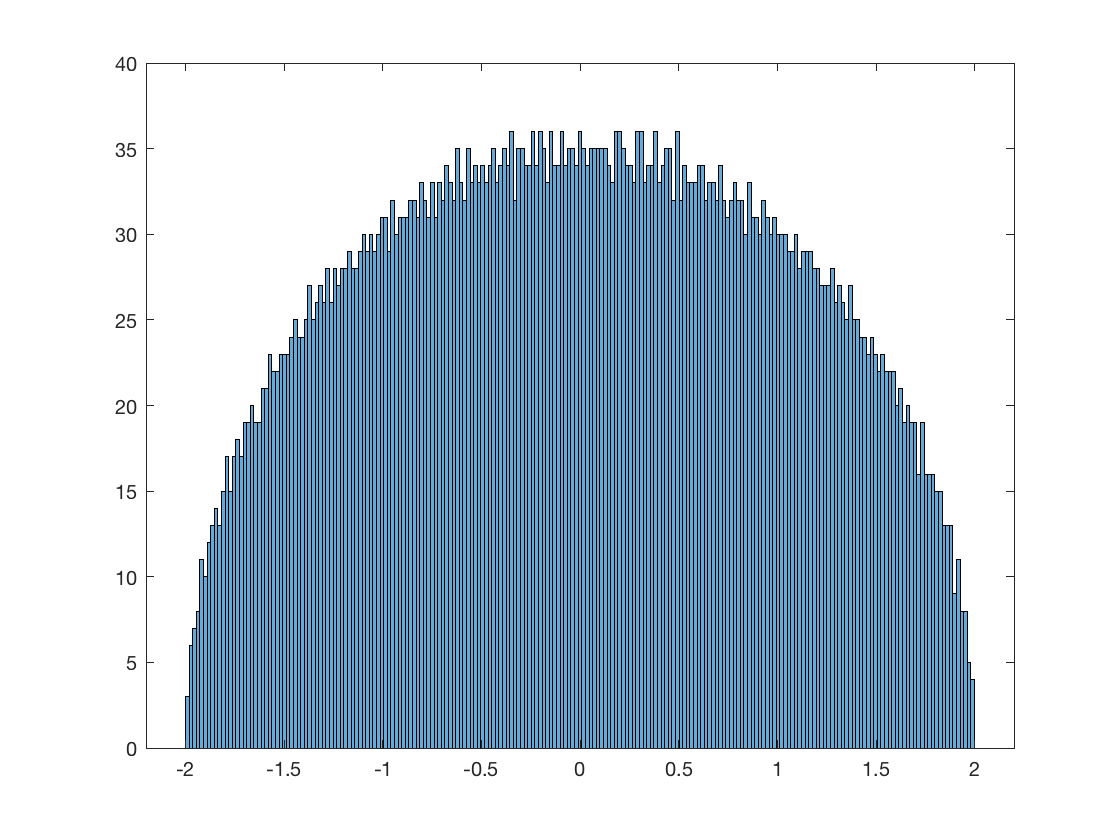}
    \caption{The empirical eigenvalue distribution $\frac{1}{N}\sum_{j=1}^N \delta_{\lam_j}$ of a $6000\times 6000$ matrix $H$ defined via \eqref{eq:Xdefn} with $\om_i=i\sqrt{2}$ and $x_i=0$ for $i\geq 1$.}
    \label{fig:sqrt2}
\end{center}
\end{figure}

The result includes the case of rectangular matrices $X$ generated from the skew-shift dynamics. In that case, we prove that the limiting eigenvalue distribution is given by the Marchenko-Pastur law \cite{MP}. Altogether, the results evidence the random-like behavior of the skew-shift. 

Our results are global in nature, i.e., they establish convergence to the limiting distributions on order $1$ scales. Extending the results to a local semicircle law in the form of \cite{BKY,BK,EKYY,ESY} and others is an open problem. 

In Section \ref{sect:deterministic}, we also discuss some numerical and analytical results concerning fully deterministic matrices. In summary, the fully deterministic situation appears to be more delicate, but nonetheless, a semicircular distribution and GUE eigenvalue spacing (Wigner surmise) can still be observed for certain models which are sufficiently quasi-random. See Table \ref{table} and Figure \ref{fig:spacing} for instances of this phenomenon. These observations suggest certain deterministic matrices that may belong to the universality class of random Hermitian matrices from the perspective of spectral statistics. Deriving such a result based on the properties of the underlying deterministic dynamical system is an open problem.  \\

We close the introduction with a brief review of two related well-known conjectures concerning the quasi-random behavior of the skew-shift.

\be{itemize}
\item 
Rudnick, Sarnak, and Zaharescu \cite{RSZ} conjectured that skew-shift orbits exhibit Poissonian spacing (as i.i.d.\ sequences would), and in fact proved this along subsequences for topologically generic frequencies; see also \cite{DRH,MS,MY,RS}. By contrast, the spacing distribution of irrational circle rotation displays level repulsion \cite{Ble,PBG}. 

\item 
In mathematical physics, the conjecture that the one-dimensional discrete Anderson model with on-site potential given by $\lam\cos\l(\binom{j}{2} \om+jy+x\r)$ exhibits Anderson localization for arbitrarily small coupling constant $\lam>0$, just like the random model \cite{Anderson,KS}, has seen only limited progress \cite{B1,BGS,HLS1,HLS2,K1,K2}, the most significant result for small $\lam$ being due to Bourgain \cite{B2}. Note that the conjecture again says that the skew-shift behaves markedly different from circle rotation $\lam\cos(j\om+x)$, which is Anderson localized if and only if $\lam>1$ \cite{Jit}.
\e{itemize}

\subsection{The model}
Let $\mathbb T$ be the one-dimensional torus, which we identify with $[0,1]$ in the usual way. For the skew-shift, the role of the ``angle'' is played by the frequency $\om \in[0,1]$. The skew-shift is then the transformation
$$
\begin{aligned}
T:&\T^2\to\T^2\\
&(x,y)\mapsto (x+y,y+\om).
\end{aligned}
$$
We write $T^j$ for the $j$-fold iteration of $T$ and $(T^j(x,y))_1$, for the first component of the vector $T^j(x,y)\in \T^2$, i.e.,
\beq\label{eq:Tjformula}
(T^j(x,y))_1=\binom{j}{2} \om+jy+x.
\eeq

We are now ready to define the matrix model we will be studying.

\be{defn}[Matrix model]
Let $M,N\geq 1$ and fix three vectors
$$
\ul{x}=(x_1,\ldots,x_M),\quad \ul{y}=(y_1,\ldots,y_M),\quad \ul{\omega}=(\omega_1,\ldots,\omega_M)\in\T^M.
$$
We define the Hermitian matrix
\beq\label{eq:Hdefn}
H_{M,N}:=\ttmatrix{0}{X}{X^*}{0},
\eeq
where $X$ is an $M\times N$ complex-valued matrix given by \eqref{eq:Xdefn}, i.e.,
$$
X_{i,j}:=N^{-1/2} e[(T^j(x,y))_1],\qquad e[t]:=\exp(2\pi \sqrt{-1} t).
$$
\e{defn}

We also introduce

\be{defn}[Averaging operation]
Given a function $f:[0,1]^M\to \R$, we define the averaging operation 
$$
\mathbb{E}_{\ul{y}} [f(\ul{y})]:=
\int_0^1\ldots \int_0^1 f(y_1,y_2,\ldots,y_M) \d y_1\ldots \d y_M.
$$
\e{defn}

Our main result is a proof of the global Wigner semicircle law, respectively the global Marchenko-Pastur law. The proof applies for any sequence of frequencies $(\omega_1,\omega_2,\ldots)$ that is ``quasi-random'' in the following sense.

\begin{defn}
\label{defn:qrandom}
Let $\de,\rho>0$. The sequence of frequencies $(\omega_1,\omega_2,\ldots)$ with $\omega_i\in [0,1]$ is $(\de,\rho)$-quasi-random, if there exists a constant $C>0$ such that
\beq
\l|\frac{1}{N^5} \sum_{i_1,i_2=1}^{\floor{\rho N}} \sum^N_{\substack{j_1,j_2,j_3,j_4=1\\ j_1+j_3=j_2+j_4}} 
e\l[\frac{\om_i-\om_{i'}}{2}(j_1^2-j_2^2+j_3^2-j_4^2)\r]\r|\leq C N^{-\de},
\eeq
for all $N\geq 1$.
\end{defn}

Notice that the sum is normalized so that the trivial bound is a constant independent of $N$. The quasi-random condition is closely related to irrationality of the $\omega_i$, as one might expect from the perspective of ergodic theory.

In Section \ref{sect:ntheory} at the end of this paper, we will provide several classes of examples of frequency sequences $(\omega_1,\omega_2,\ldots)$ that are quasi-random in this sense, as well as graphs displaying the results of numerical simulations. As we will see there, the choice of $\rho$ is insignificant for verifying Definition \ref{defn:qrandom} in explicit examples.

\subsection{Main results}
We use the moment method to identify the global distribution of the eigenvalues.

Our first result, Theorem \ref{thm:keyrho} below, computes the expectation values of even moments of $H$ asymptotically as $M,N\to\infty$ with $M=\floor{\rho N}$. To state the result, we introduce the normalized moments
$$
\mu^{(2k)}_{M,N}:=\frac{1}{2N}\mathbb \Tr[H^{2k}_{M,N}].
$$
Notice that the odd moments are automatically zero, due to the block structure of $H_{M,N}$. 

Our first main result concerns the case of quadratic matrices, $M=N$. 

\be{thm}[Main result for $\rho=1$]
\label{thm:keyrho1}
Let $\de>0$ and let $(\om_1,\om_2,\ldots)$ be $(\de,1)$-quasi-random. Let $k\geq 1$. As $N\to\infty$, it holds that
\beq
\mathbb{E}_{\ul{y}}[\mu^{(2k)}_{N,N}]=c_{k}+O(N^{-\de/16}),
\eeq
where $c_k=\frac{1}{k+1}\binom{2k}{k}$ are the Catalan numbers. The estimate holds uniformly in the choice of $(x_1,x_2,\ldots)$.
\e{thm}

Our second main result concerns the case of rectangular matrices. $M=\floor{\rho N}$ with $\rho>0$.  For this, we introduce the following rescaling of the Marchenko-Pastur law
\beq 
f_{\rho^{-1}}(t):=\rho^{-2} f^{MP}_{\rho^{-1}}\l(\frac{t}{\rho}\r),
\eeq 
where the Marchenko-Pastur law with parameter $\rho^{-1}$ is given by
$$
f^{MP}_{\rho^{-1}}(t)=\frac{\sqrt{(\lam_+-t)(t-\lam_-)}}{2\pi t \rho^{-1}}\ind_{\lam_-\leq t\leq \lam_+},\qquad \lam_{\pm}=(1\pm \rho^{-1/2})^2.
$$
We write $\mu_k$ for the moments of $f_{\rho^{-1}}$.

 \be{thm}[Main result for general $\rho>0$] 
\label{thm:keyrho}
For $\rho>0$, set $M=\floor{\rho N}$. Let $\de>0$ and let $(\om_1,\om_2,\ldots)$ be $(\de,\rho)$-quasi-random. Let $k\geq 1$ and $\rho >0$. As $N\to\infty$, it holds that
\beq
\mathbb{E}_{\ul{y}}[\mu^{(2k)}_{M,N}]=\mu_k+O(N^{-\de/16}).
\eeq
The estimate holds uniformly in the choice of $(x_1,x_2,\ldots)$.
\e{thm}

\be{cor}\label{cor:distr}
From Theorems \ref{thm:keyrho1} and \ref{thm:keyrho}, we obtain that the empirical spectral distribution $\frac{1}{N}\sum_{j=1}^N \delta_{\lam_j}$
converges weakly in distribution to the Wigner semicircle law, respectively, the Marchenko-Pastur law (depending on the value of $\rho$).
\e{cor}

\begin{proof}
The corollary follows from the solvability of the associated moment problem, since both distributions are compactly supported.
\end{proof}

In Section \ref{sect:deterministic}, we discuss examples of fully deterministic matrices where no average $\mathbb E_{\ul{y}}$ is taken, and instead $\ul{x}=\ul{y}=0$. Our findings there show that the situation is delicate in the deterministic class: Numerically, we observe that the global eigenvalue distribution is semicircular only for some models (e.g., $\om_i=\sqrt{i}$) which are sufficiently quasi-random, but not for other ones which involve linear terms, like $\om_i=i\sqrt{2}$. (Note that this bears some similarity with the situation in other contexts mentioned at the end of the introduction.) Moreover, even when the distribution is semicircular, it is accompanied by heavy tails which render the moment method ineffective. See Table \ref{table} for a summary. We also establish an analytical bound for the moment of a deterministic model (Theorem \ref{thm:modelb}) and we observe numerically that the eigenvalue spacing for $\om_i=\sqrt{i}$ (and similar models) matches that of GUE matrices (Wigner surmise).

\subsection{Some possible extensions}
The main results and their proofs extend verbatim if $y_1,\ldots,y_M$ are not sampled uniformly, but with respect to another measure $\d\mu$ on $[0,1]$ such that $\int_0^1 e[jy]\d\mu(y)=\de_{j,0}$.

Moreover, one can instead study matrices 
\beq\label{eq:jpmodel}
X_{i,j}=e[ j^p\om_i+j y_i]
\eeq
with $p$ an integer $>1$ and $y_1,\ldots,y_M\in [0,1]$ again sampled uniformly and independently at random. Alternatively, if one wants to generate the matrix elements again as orbits of a true dynamical system, one can take the skew-shift on the $p$-torus. In either of these cases, the method we develop here applies. Of course, the relevant input about exponential sums (the $(\de,\rho)$-quasirandom condition for our skew-shift model) changes from case to case. We note that the presence of the $y_i$ in \eqref{eq:jpmodel} is crucial for our method. The reason is that the average $\mathbb E_{\ul{y}}$ ensures the validity of the Kirchhoff circuit law in our graphical representation for the moments (see the next section). Without the Kirchhoff circuit law, e.g., in the deterministic setting, one needs to understand the exponential sums much more precisely to derive the semicircle law.

We also remark that in light of the ergodic theorem for the skew-shift dynamics, it may be possible to strengthen the convergence in Corollary \ref{cor:distr} to an almost sure result, but we do not dwell on this here.\\ 

\textbf{Notation.}
We write $C>0$ for universal constants. The proof of Theorems \ref{thm:keyrho1} and \ref{thm:keyrho} is given for a fixed choice of $k$ and $\rho$. We write $C_k$ for a constant that may depend on $k$ and $\rho$, but on $N$, and whose value may change from line to line.

\section{Graphical representation of the moments}
In this short section, we relate the expected moments $\E_{\ul{y}}[\mu^{(2k)}_{M,N}]$ to a sum over graphs with edge weights.\\

To this end, we introduce a  notion of ``exploration graph''. In a nutshell, an exploration graph is a directed graph that is generated by following a single closed path. (Multiple edges between each pair of vertices are allowed; this includes self-loops.) 

\begin{defn}[Explorations and exploration graphs]
\label{defn:exploration}
  Let $1\leq l\leq k$.
  \begin{enumerate}[label=(\roman*)]
\item An \emph{exploration} on $k$ edges and $l$ vertices is a list $L\in (\{1,2,\ldots,l\}^2)^k$ of the form
   $$
 L=((\nu_1,\nu_2),(\nu_2,\nu_3),\ldots,(\nu_{k-1},\nu_k),(\nu_k,\nu_1))
 $$
 where the numbers $\nu_1,\ldots,\nu_k\in \{1,2,\ldots,l\}$ satisfy the following two conditions:
 \begin{itemize}
     \item $\{\nu_1,\ldots,\nu_k\}=\{1,\ldots,l\}$.
     \item For all $1\leq j\leq l-1$, we have
 $$
 \min\setof{1\leq i\leq k}{\nu_i=j}<  \min\setof{1\leq i\leq k}{\nu_i=j+1},
 $$
 i.e., the first label $j$ occurs before the first label $j+1$. 
 \end{itemize}
 
 \item Each exploration $L$ defines an ``exploration graph'' $G_L=(V,L)$ as follows. One takes $V=\{\nu_1,\ldots,\nu_k\}$ as the vertex set and the elements of $L$ as the set of directed edges of $G_L$. The set of directed edges inherits an order from $L$. We write $\curly{L}_k$ for the set of exploration graphs on $k$ edges. 
 
 \item To any list $\ul{i}=(i_1,i_2,\ldots,i_k)\subset \{1,\ldots,M\}^k$, we associate a list of edges
$$
L_{\ul{i}}:=((i_1,i_2),(i_2,i_3),(i_3,i_4),\ldots,(i_k,i_1)).
$$
Let $l=|\{i_1,\ldots,i_k\}|$ and let $L$ be an exploration on $l$ vertices. We write
$$
\begin{aligned}
L_{\ul{i}}\sim L \quad\Longleftrightarrow\quad \exists\, &\textnormal{bijection}\, 
\sigma:\{i_1,\ldots,i_k\}\to\{1_,\ldots,l\}\\
&\textnormal{ such that } L_{(\sigma(i_1),\ldots,\sigma(i_k))}=L.
\end{aligned}
$$
 \end{enumerate}
\end{defn}

Notice that the exploration $L$ generates a closed path on the exploration graph $G_L$, and this path is by construction an Eulerian circuit (meaning it visits every edge exactly once). As a consequence, every vertex of $G_L$ has the same in-degree as out-degree.

\be{rmk}
The point of (iii) is that $L_{\ul{i}}\sim L$  holds iff the two lists correspond to the same exploration (when vertex labeling is ignored so that the vertices are exactly $\{1,\ldots,l\}$). 
Pictorially, $L_{\ul{i}}\sim L$ means that $L_{\ul{i}}$ and $L$ lead to the same graph if the order in which edges are traversed is kept.
\e{rmk}

The exploration graphs will be endowed with integer-valued edge weights (or ``currents'') satisfying Kirchhoff's current law.

\begin{defn}[Edge weights]\label{defn:Feynman}
Let $L$ be an exploration and $G_L=(V,L)$ its associated exploration graph. Given a vertex $v\in V$, we write $O_v$ for the set of outgoing edges from $v$, and $I_v$ for the set of incoming edges. Given a sequence $\ul{j}=(j_1,\ldots, j_k)\subset \{1,\ldots,N\}^k$, we assign the weight $j_i$ to the edge $(\nu_i,\nu_{i+1})$ in $L$.

We say that the sequence $\ul{j}$ is an admissible collection of edge weights for $L$ (or ``$L$-admissible'' for short), if the Kirchhoff circuit law holds on $G_L$, i.e., if
$$
\sum_{e \in I_v} j_e = \sum_{e \in O_v} j_e,\qquad \forall v\in V.
$$
\end{defn}



With these graph-theoretic notions at hand, we can write down a graphical representation formula for the moments $\mathbb{E}_{\ul{y}}[\mu^{(2k)}_{M,N}]$. This formula is the starting point for the  subsequent analysis.\\

\paragraph{\textbf{Notation.}} We always use $1\leq i\leq M$ for row indices and $1\leq j\leq N$ for column indices. I.e., a sum $\sum_{\ul{i}=(i_1,\ldots,i_k)}$ is implicitly taken over $\ul{i}\in \{1,\ldots,M\}$ (with $M=\floor{\rho N}$) and a sum $\sum_{\ul{j}=(j_1,\ldots,j_k)}$ is implicitly taken over $\ul{j}\in \{1,\ldots,N\}$, unless specified otherwise. Moreover, we identify $i_{k+1}=i_1$. Borrowing convenient physics terminology, we call
\beq\label{Kpropdefn}
K_{(i,i')}(j):= e\l[(\om_i-\om_{i'})\frac{j^2}{2}\r]
\eeq
the ``effective propagator (from $i$ to $i'$ at momentum $j$)''. We recall that $\curly{L}_k$ is the set of all exploration graphs on $k$ edges (and therefore $\leq k$ vertices).

\be{prop}[Graphical representation formula for the moments]\label{prop:graphrep}
We have
\beq
\mathbb{E}_{\ul{y}}[\mu^{(2k)}_{M,N}]
=\frac{1}{N^{1+k}}\sum_{G_L=(V,L)\in \curly{L}_k}  \sum_{\substack{\ul{i}=(i_1,\ldots,i_k)\\ L_{\ul{i}}\sim L}}\sum_{\substack{\ul{j}=(j_1,\ldots,j_k)\\ \ul{j}\textnormal{ is $L$-admissible}}} \prod_{r=1}^k K_{(i_r,i_{r+1})}(j_r)
\eeq
\e{prop}

We mention that the collective sum over exploration graphs and edge weights can also be viewed as a sum over Feynman graphs, albeit with a spatially dependent propagator.

\be{proof}
We compute the moment as
$$
\begin{aligned}
\mu^{(2k)}_{M,N}
=&\frac{1}{2N}\Tr[H_{M,N}^{2k}]
=\frac{1}{N}\Tr[(XX^*)^k]
=\frac{1}{N}\sum_{i_1,\ldots,i_k=1}^M \sum_{j_1,\ldots,j_k=1}^N  \prod_{l=1}^k X_{i_l,j_l}X^*_{j_l,i_{l+1}}
\end{aligned}
$$
The true propagator (before averaging) is
\beq\label{eq:Ktilde}
\begin{aligned}
\tilde K_{(i,i')}(j):= N X_{i,j}X^*_{j,i'}=&e\l[\frac{\om_i-\om_{i'}}{2} (j^2-j)+(y_{i}-y_{i'})j+(x_{i}-x_{i'})\r]\\
=&e\l[\frac{\om_i-\om_{i'}}{2} j^2+(y_{i}-\om_i-(y_{i'}-\om_{i'}))j+(x_{i}-x_{i'})\r]
\end{aligned}
\eeq
We can write $\mu^{(2k)}_{M,N}$ as a sum over exploration graphs on $k$ edges:
\beq\label{eq:graphical0}
\mu^{(2k)}_{M,N}=\frac{1}{N^{1+k}}\sum_{G_L=(V,L)\in \curly{L}_k}  \sum_{\substack{\ul{i}=(i_1,\ldots,i_k)\\ L_{\ul{i}}\sim L}}\sum_{\ul{j}=(j_1,\ldots,j_k)} \prod_{r=1}^k \tilde K_{(i_r,i_{r+1})}(j_r).
\eeq
We first note that the phases $x_i$ telescope to zero along the exploration, meaning that $\sum_{r=1}^k (x_{i_r}-x_{i_{r+1}})=0$, since $i_{k+1}=i_1$.

The claim of the proposition can then be restated as saying that taking the average $\mathbb{E}_{\ul{y}}$ on both sides of \eqref{eq:graphical0} has two effects: (a) it retains only $L$-admissible edge weights $\ul{j}$ (i.e., it enforces the Kirchhoff circuit law at each vertex) and (b) it replaces the true propagator $\tilde K$ by the effective propagator $K$.

The fact that (a) and (b) hold follows directly from the formula \eqref{eq:Ktilde} and orthogonality of the functions $\{e[j\cdot]\}_{j\in \Z}$ over the torus $\T$, by a straightforward computation. This proves Proposition \ref{prop:graphrep}.
\e{proof}


\section{Characterization of subleading graphs}
\label{sect:graphtheory}

In this section, we will work with ordinary graphs $G=(V,E)$ with each vertex having even degree and undirected edges.

\subsection{Preprocessing and good cycles}

We define a notion of preprocessing which simplifies a graph without significantly changing the moment sum in Proposition \ref{prop:graphrep}.

\begin{figure}[t]
\begin{center}
    \includegraphics[scale=.4]{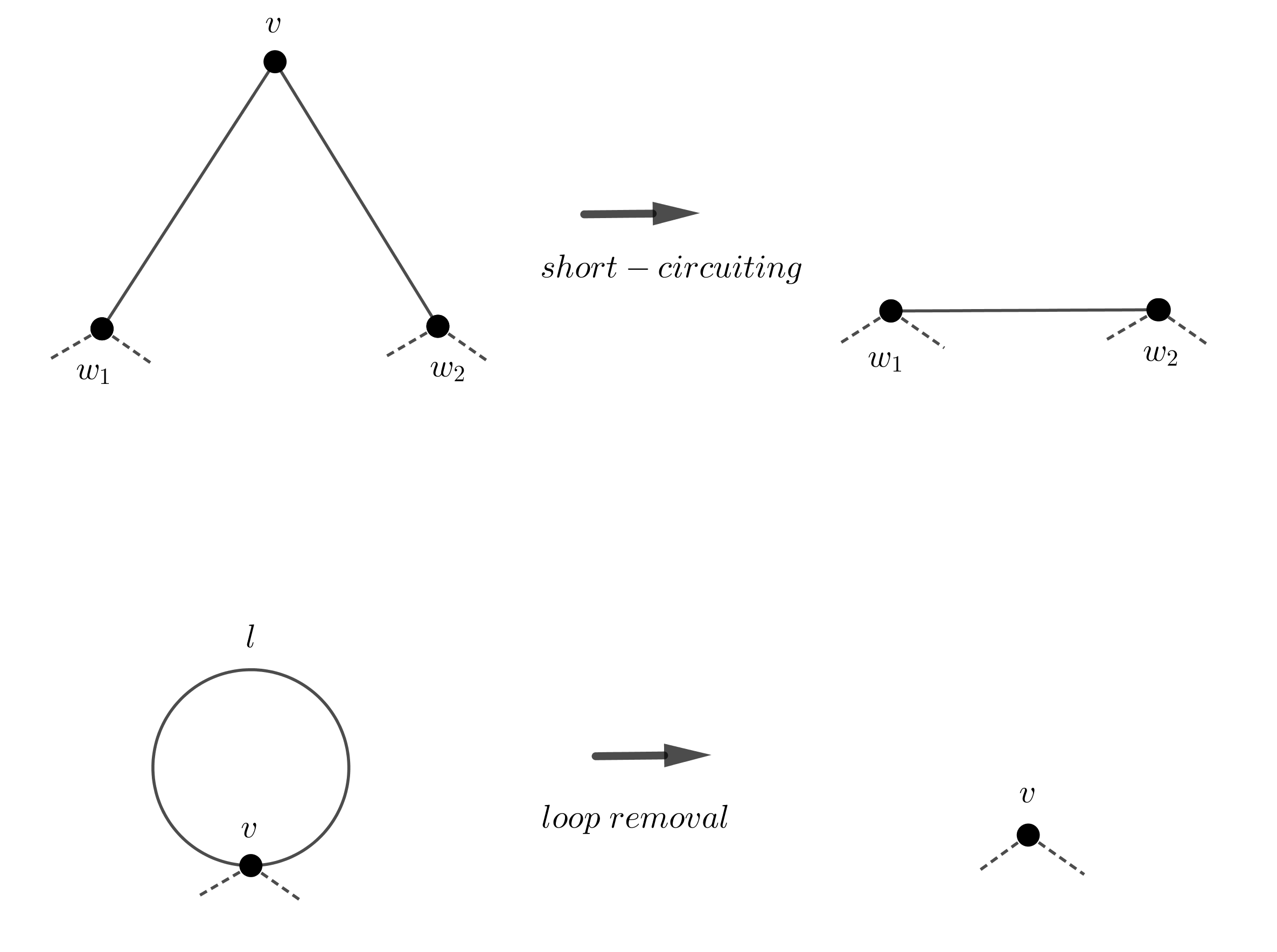}
    \caption{The two types of preprocessing}
\end{center}
\end{figure}

\begin{defn}[Preprocessing]\label{defn:preprocessing}
Consider a graph $G=(V,E)$ such that each vertex has even degree. We will iteratively apply the following two processes when possible.

\begin{enumerate}[label=(\roman*)]
    \item  
    Short-circuiting: If the graph $G$ has a vertex $v$ such that $v$ has only 2 edges $(v,w_1)$ and $(v,w_2)$, then the graph $\curly{S}(G,v)$ with the vertex $v$ short-circuited is defined as follows. From the graph $G$ we remove $v$ and its adjacent two edges and finally we replace them by the edge $(w_1,w_2)$
    
    \item 
    Loop removal: If the graph $G$ has a self-loop, $l$, at the vertex $v$, then $\curly{L}(G,l)$ is the graph $G$ with loop $l$ removed.
\end{enumerate}
A fully preprocessed graph is a graph $H$ upon which no preprocessing step can be applied.
\end{defn}

The structure we identify for characterizing leading versus subleading graphs is the following ``good cycle''.

\begin{figure}[t]
\begin{center}
    \includegraphics[scale=.4]{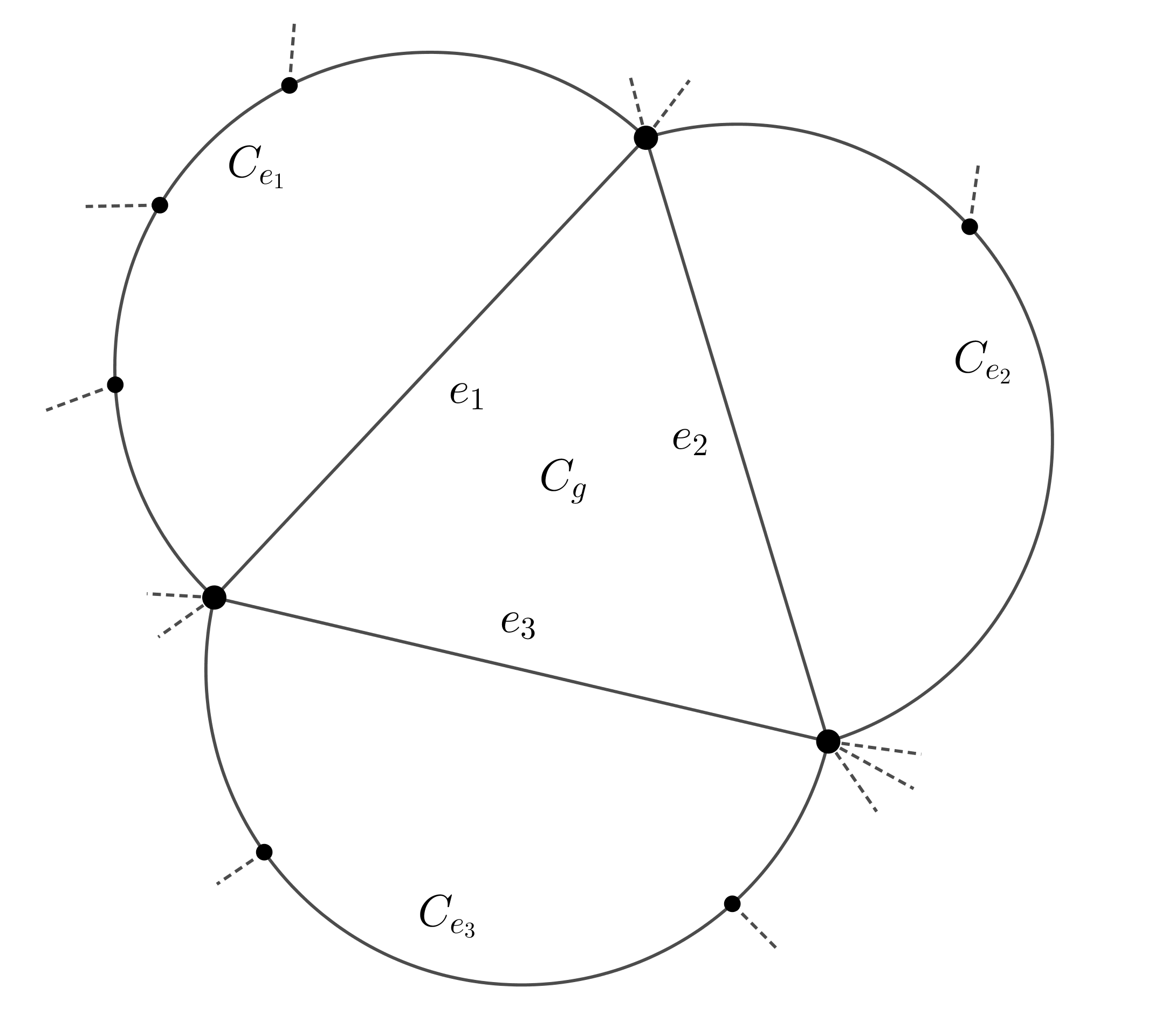}
\end{center}
\caption{An example of a good cycle $C_g=e_1\to e_2\to e_3\to e_1$. The semicircles represent the cycles $C_{e_i}$ which intersect $C_g$ only at $e_i$, and the dashed lines represent their connections to the rest of the graph.}
\end{figure}

\begin{defn}[Good cycle]\label{defn:good}
A good cycle $C$ is a simple cycle such that for every edge $e \in C$ there exists a cycle $C_{e}$ such that $C_{e}\cap C = \{e\}$
\end{defn}

The main result of this section is the following theorem which establishes the existence of a good cycle. In Section \ref{sect:control}, we will show that graphs containing a good cycle are subleading.

\begin{thm}[Existence of good cycle]
\label{thm:good}
If a fully preprocessed graph is not a point, then it has a good cycle.
\end{thm}

We will prove this theorem by contradiction. From now on, we assume for a contradiction that there exist non-trivial fully preprocessed graphs that have no good cycle. We fix a minimal graph of this kind, call it $G$, where minimal is defined as having the minimal number of vertices. In the following, we will refer to $G$ as the ``smallest counterexample''.

\subsection{A good cycle in the contracted graph $G/ \tilde e$}

\begin{defn}
Fix an edge $\tilde{e}= (\tilde{v},\tilde{w})$ in $G$. The contracted graph $G / \tilde{e}$ is defined by contracting the edge $\tilde{e}$, which means combining the two vertices $\tilde{v}$ and $\tilde{w}$ into a new vertex $v_c$ in $G / \tilde{e}$, and replacing any additional edges between $\tilde v$ and $\tilde w$ by self-loops at $v_c$.
\end{defn}

\begin{lm}\label{lm:fullyprep}
 $G / \tilde{e}$ is a fully preprocessed graph.
\end{lm}
\begin{proof}
We will perform some case analysis based on the number of edges between $\tilde{v}$ and $\tilde{w}$ in $G$.

\textit{Case 1:} There are at least $3$ edges between $\tilde{v}$ and $\tilde{w}$. In this case, we claim that the two edge cycle consisting of any two edges $e_1,e_2$ between $v$ and $w$ will be a good cycle in G. Let $e_3$ be a third edge between $v$ and $w$, then $e_1 \cup e_3$ will be a cycle that intersects $e_1 \cup e_2$ in only $e_1$ while $e_2 \cup e_3$ will be a cycle that intersects $e_1 \cup e_2$ in only $e_2$. Clearly, then $e_1 \cup e_2$ is a good cycle in G, which is a contradiction to the choice of $G$.

\textit{Case 2:} Now we consider the case that there are exactly two edges $e_1,e_2$ in between the vertices $\tilde{v}$ and $\tilde{w}$. If $G \setminus \{e_1,e_2\}$ is a connected graph, then there is a path, $p$, between $\tilde{v}$ and $\tilde{w}$ that does not use either of the edges $\tilde{v}$ and $\tilde{w}$. We can then argue that $e_1 \cup e_2$ is a good cycle. Indeed, the good cycle conditions are verified by the two cycles $e_1 \cup p$ and $e_2 \cup p$.

If the graph disconnects upon removing the edges $e_1:= \tilde{e}$ and $e_2$, then one can check that $G / \tilde{e}$ will have one self-loop $l$ at $v_c$ and the graph $\mathcal{L}(G / \tilde{e},l)$ will need no preprocessing; this is because all vertices in $\mathcal{L}(G / \tilde{e},l)$ will have degree at least 4. Since $G$ was a smallest counterexample, there necessarily exists a good cycle $C_{\text{good}}$ in $\mathcal{L}(G / \tilde{e},l)$. (Notice that already $G/\tilde e$ has strictly fewer vertices than $G$.) Since $C_{\text{good}}$ is simple, one can check that its lift to $G$ can be chosen such that it uses neither of the edges $e_1$ or $e_2$. This would imply that $G$ has a good cycle, which is a contradiction.

\textit{Case 3:} There is only 1 edge in between $\tilde{v}$ and $\tilde{w}$. Then $G / \tilde{e}$ will have no self-loops and all vertices in $G / \tilde{e}$ will have degree at least 4. No further preprocessing steps need to be taken.
\end{proof}

Now, we combine Lemma \ref{lm:fullyprep} with the fact that $G$ is the smallest counterexample to find a good cycle in $G / \tilde{e} $. (Notice that $G/\tilde e$ has strictly fewer vertices than $G$.) Call $C_{\text{good}}$ the good cycle in $G/ \tilde{e}$. Define $\hat{C}$ as the lift of $C_
{\text{good}}$ to the graph $G$. Notice that if the cycle $\hat{C}$ does not use the edge $\tilde{e}$, then $\hat{C}$ will be a good cycle in $G$ and $G$ cannot be the smallest counterexample. 

\subsection{Type I and type II edges}

We now delve into the major case of our analysis, where $\hat{C}$ contains the edge $\tilde{e}$.

We first divide the edges of the cycle $\hat{C}$ into two types.
\begin{defn}[Type I and II edges]\label{def:edgetyp}
Consider the graph $G$ and the lifted good cycle $\hat{C}$ as constructed earlier.
\begin{enumerate}
    \item Type I edge: An edge $e_1 \in \hat{C}$ with $e\neq \tilde e$ will be considered a type I edge if there exists a cycle $\hat C_{e_1}$ in $G$ such that $\hat C_{e_1} \cap \hat{C} = \{ e_1 \}$
    \item Type II edge: A type II edge is an edge $e_2\in\hat C$, $e_2\neq \tilde e$ that is not type I.
\end{enumerate}
\end{defn}
\begin{figure}
\begin{center}
    \includegraphics[scale=0.5]{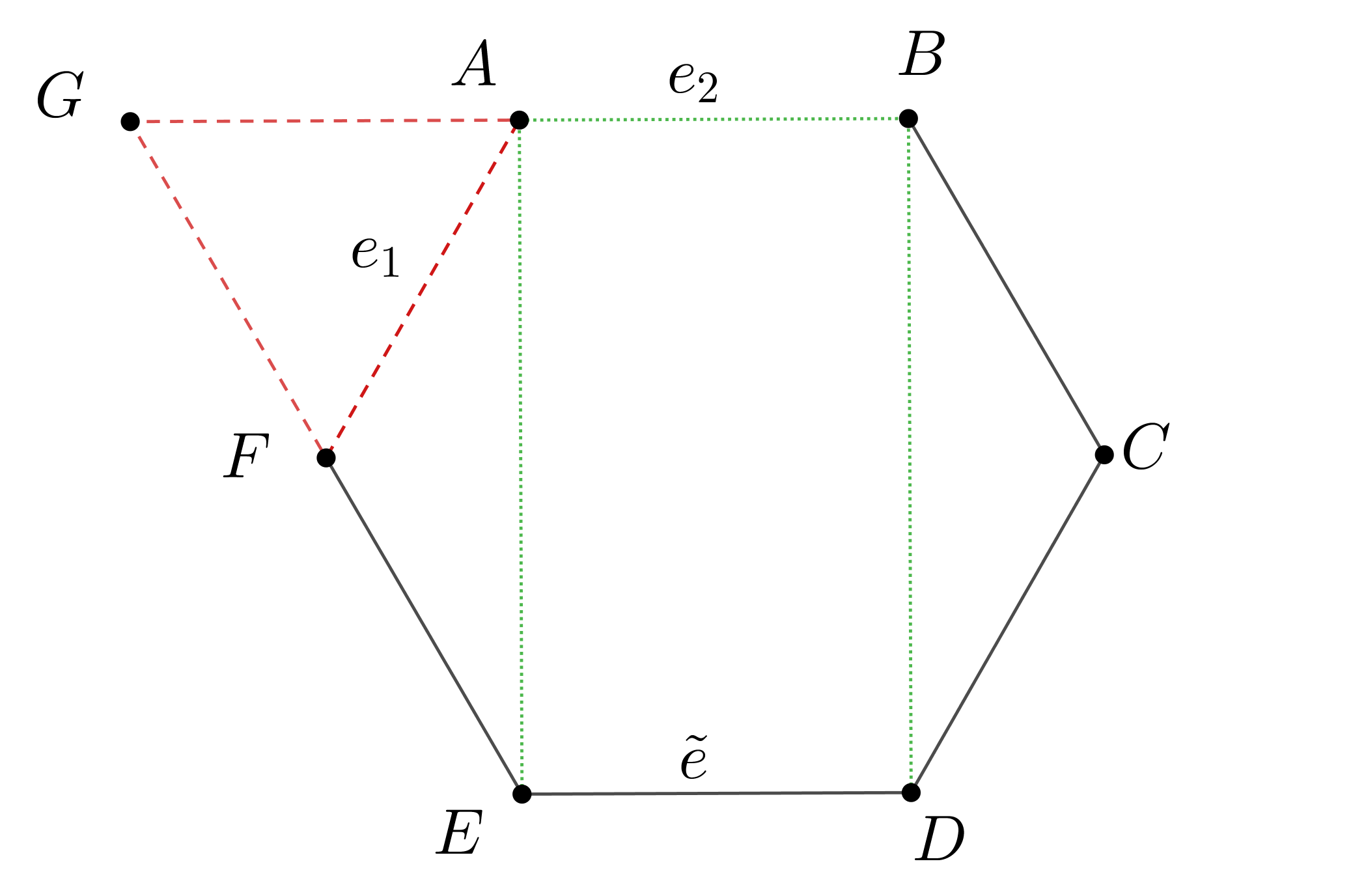}
\end{center}
\caption{In this example of a cycle $\{A,B,C,D,E,F\}$, the red edge $e_1$ is a type I edge, while the green edge $e_2$ is a type II edge.}
\end{figure}

\begin{lm}\label{lm:typeIIcycle}
For any  type II edge $e_2$ there exists a cycle $\hat C_{e_2}$ in $G$ such that $\hat C_{e_2} \cap \hat{C} = \{e_2,\tilde{e} \}$.
\end{lm}

\begin{proof}
Consider a type II edge $e_2$.
We know that since $C_{\text{good}}$ was a good cycle in $G / \tilde{e}$, there exists a cycle $C_{e_2}$ whose intersection with $C_{\text{good}}$ is only $e_2$. When we lift these two cycles to $G$, the only possible change is that we may add the removed edge $\tilde{e}$. 

Letting $\hat C_{e_2}$ be the lift of $C_{e_2}$ to $G$, we see that the intersection of $\hat C_{e_2}$ with $\hat{C}$ could either be $\{e_2\}$ or $\{\tilde{e}, e_2\}$. If it were the former, then we would say $e_2$ is a type I edge. Since $e_2$ is a type II edge, it must be the latter.
\end{proof}

We can engage in casework depending on whether there are an even or an odd number of edges of type II. We recall that our overall goal is to derive a contradiction (to the assumption that a smallest counterexample $G$ exists).

\subsection{Excluding an even number of type II edges}

\begin{prop}\label{prop:typeIIodd}
The assumption that there is an even number of type II edges leads to a contradiction.
\end{prop}

\begin{proof}[Proof of Proposition \ref{prop:typeIIodd}]
We will show that if there is an even number of type II edges, then $\hat{C}$ is a good cycle, which contradicts the choice of $G$. 

Let $\hat{e}$ be an arbitrary type II edge in $\hat{C}$. For every edge $e\in \hat C$, we choose an associated cycle $\hat C_e$. (For type I edges, we choose $\hat C_e$ from the definition, for type II edges, we choose it by Lemma \ref{lm:typeIIcycle}.)

Start with the union of cycles 
\begin{equation}
U_{\hat{e}}:= \hat{C} \cup\bigcup_{\substack{e\in \hat{C}:\\ e\ne \hat e,\tilde{e}}} \hat C_{e}
\end{equation}

From this, we want to construct a cycle $\hat{B}_{\hat{e}}$ whose intersection with $\hat{C}$ will only be the edge $\hat{e}$. For every edge $e \ne \hat{e}, \tilde{e}$ in $\hat{C}$ we will have two appearances in $U_{\hat{e}}$, once in $\hat{C}$ and once in $C_{e}$. An application of the proof of the bypass lemma \ref{lem:bypass} will show that we can construct a cycle, not necessarily simple, without an appearance of the edge $e$ for all $e \ne \hat{e}$ or $\tilde{e}$. One can see from the proof that the edge $\tilde{e}$ will appear an even number of times in such a cycle while $\hat{e}$ will appear only once.

By the even bypass lemma \ref{lem:EvBypass}, we obtain a cycle whose only intersection with $\hat{C}$ is $\hat{e}$, but which may not be simple.

There is a simple procedure which one can call ``loop erasure" for turning any cycle that is not simple into a non-trivial simple cycle. Start from the edge $e$ and go along the cycle. When reaching a vertex that is used twice or more, simply cut out the part of the cycle that occurs between its first and final appearance. Eventually, this procedure will result in a simple cycle. Since the edge $\hat e$ was contained only once in the original cycle, the resulting simple cycle cannot be the trivial cycle that traverses the edge $\hat e$ twice back and forth. It is a non-trivial simple cycle containing $\hat e$. This cycle establishes that $\hat e$ is a type I edge, a contradiction to the assumption that it is a type II edge. Hence, the only possible even number of type II edges is zero. However, this implies that all edges are type I and $\hat C$ is a good cycle in $G$, a contradiction to the choice of $G$.
\end{proof}
%
%
From now on we assume that there is an odd number of edges of type II.

\subsection{Excluding an odd number of type II edges}

We now consider the complementary case to Proposition \ref{prop:typeIIodd}.

\begin{prop}\label{prop:typeIIeven}
The assumption that there is an odd number of type II edges leads to a contradiction.
\end{prop}

Notice that, taken together, Propositions \ref{prop:typeIIodd} and \ref{prop:typeIIeven} lead to a contradiction based solely on the existence of the smallest counterexample $G$, and hence establish Theorem \ref{thm:good}.

In the remainder of this section, we will prove Proposition \ref{prop:typeIIeven}. From now on, we assume that there are an odd number of type II edges. We divide their endpoints into two categories --- positive and negative ones.
\begin{defn}[Positive and negative vertices]\label{def:Sign}
In order to properly define this notion, we need to give an orientation to the cycle $\hat{C}$, taken to start at the special edge $\tilde{e}=(\tilde v,\tilde w)$. Letting $e_0:=\tilde{e}$, we define the oriented version of the cycle $\hat{C}$ by
\begin{equation}
    \hat{C}_{\to} := e_0 \rightarrow e_1 \rightarrow e_2 \rightarrow ... \rightarrow e_m \rightarrow e_0
\end{equation}
Given this ordering of the edges along the cycle, we can order the type II edges $\hat{e}_1,\ldots,\hat{e}_k$ according to the order in which they are visited by  $\hat{C}_\to$. Namely, if $e_{i_1},e_{i_2},\ldots,e_{i_k}$ are the edges of type II on the cycle $\hat{C}_\to$ with $i_1<i_2<\ldots<i_k$, then we set $\hat{e}_l = e_{i_l}$. We will also let $\hat{e}_0$ be our special edge $\tilde{e}$ and, for the purposes of this definition, we consider it as an edge of type II. Each edge $\hat{e}_l$ will inherit its orientation from the one assigned to $\hat{C}$; we can write each edge with its orientation as $\hat{e}_l = \hat{v}_l \rightarrow \hat{w}_l$. 

We can now define the set of positive vertices, $B_{+}$, and negative vertices, $B_{-}$ 
\begin{align}
    &B_{+}:= \bigcup_{l= 0 \text{ mod } 2} \hat{v}_l \cup \bigcup_{l = 1 \text{ mod } 2} \hat{w}_l\\
    &B_{-}:= \bigcup_{l=0 \text{ mod } 2} \hat{w}_l \cup \bigcup_{l =1 \text{ mod } 2} \hat{v}_l
\end{align}
Moreover, we define $\tilde B_+$ as the vertices in the good cycle $\hat C_{\to}$ that lie in between the $B_+$-vertices, and analogously, we define $\tilde B_-$ as the vertices that lie in between the $B_-$ vertices.
\end{defn}

\begin{figure}
    \centering
    \includegraphics[scale=.35]{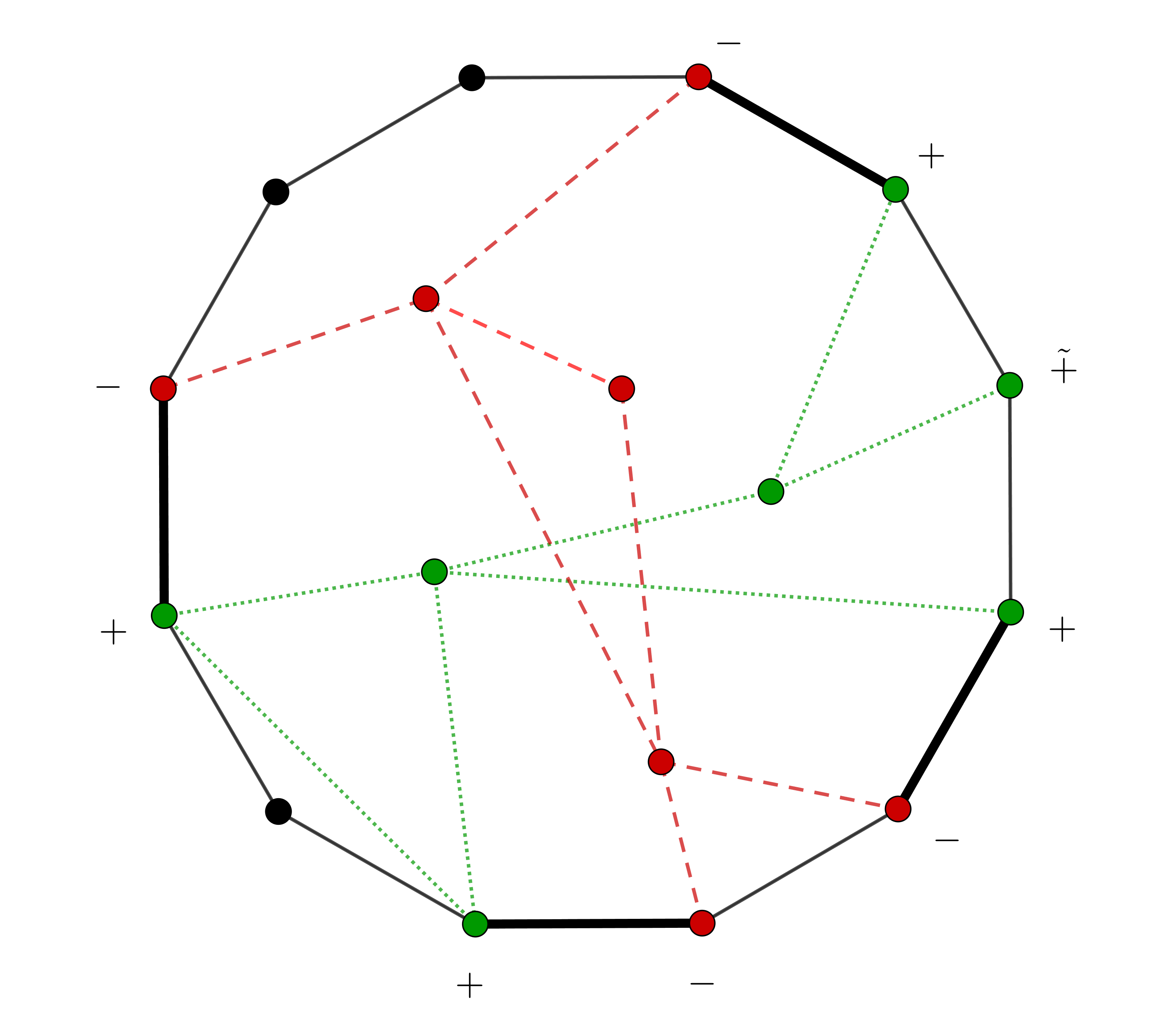}
    \caption{An example of the vertex decomposition into $+$ and $-$ vertices from Definition \ref{def:Sign}. The bold edges are the type $II$ edges in this cycle. The green $+$ vertices on the boundary form the set $B_+$. The graph $G_+$ is constructed from all the internal vertices connected to $B_+$ (and analogously for $B_-$ and $G_-$). There is a $\tilde +$ vertex, which also contributes to the set $G_+$.}
    \label{fig:PositiveNegVer}
\end{figure}

See Figure \ref{fig:PositiveNegVer} for an example of this definition. We remark that the definition uses the fact that we have an odd number of edges of type II. There will also be no problem with our assignment if edges of type II happen to be adjacent to each other.

We will now distinguish various cases concerning $B_\pm$. We make the disclaimer that we will use the symbol ``$\setminus$'' on graphs with two different meanings, either for the removal of vertices, or for the removal of edges, which we believe are clear from context. In particular, $G\setminus \hat{C}$ will mean the graph $G$ with the edges of $\hat{C}$ removed.\\

\textit{Case 1:}  $B_+$ \textit{and} $B_-$ \textit{are connected as subgraphs of } $G\setminus \hat{C}$. We will see that this case is very similar to the case in which we had an odd number of edges of type II. Namely, let $v_{+} \in B_{+}$ be connected to $v_{-} \in B_{-}$; let $p_{v_+,v_-}$ be the path in $G$ between $v_{+}$ and $v_{-}$ not using the edges of the cycle $\hat{C}$ and, for each $1\leq i\leq k$, let $p^{i}_{v_+,v_{-}}$ be the part of the cycle $\hat{C}$ between the vertices $v_{+}$ and $v_{-}$ that contains the edge $\hat{e}_i$.

For every edge $\hat{e}_i$ of type II, we will use this information to construct a cycle $\tilde{C}_{e_i}$ such that $\hat{C}\cap \tilde{C}_{e_i}=\{\hat{e}_i\}$. (This then leads to a contradiction in the same way as in the proof of Proposition \ref{prop:typeIIodd}.) We define the cycle $C^{i} = p_{v_+,v_-} \cup p^{i}_{v_+,v_-}$. We will apply the bypass lemma \ref{lem:bypass} to the union 
$$
\hat U_{e_i}=C^{i} \cup\bigcup_{\substack{e \in p^{i}_{v_+,v_-}\\ e\ne \tilde{e}, \hat{e}_i}} \hat C_{e},
$$
where, for type I edges $e$, we choose $\hat C_e$ from the definition, and for type II edges $e$, we choose it by Lemma \ref{lm:typeIIcycle}.

Notice that by Definition \ref{def:Sign} we must necessarily have an odd number of $\hat{e}_1,\ldots, \hat e_k$ in $p^{i}_{v_+, v_-}$. We apply the bypass lemma \ref{lem:bypass} and use each cycle $\hat C_{e}$ in $\hat U_{e_i}$ to bypass the edge $e\neq \tilde{e}, \hat{e}_i$. This results in a (possibly non-simple) cycle in which the edge $\tilde{e}$ appears an even number of times, which we can then reduce to a simple cycle, our desired $\tilde{C}_{e_i}$ by loop erasure such that $\hat{C}\cap \tilde{C}_{e_i}=\{\hat{e}_i\}$. (Here we use in particular that $\tilde e$ does not appear in $\tilde{C}_{e_i}$.) As mentioned above, since $\hat e_i$ was an arbitrary type II edge, this leads to a contradiction and finishes the case 1.
\\

\textit{Case 2:} $B_+$ and $B_-$ \textit{are disconnected as subgraphs of} $G\setminus \hat{C}$. 

\begin{defn}
We define $G_\pm$ as the subgraph of $G\setminus \hat C$ containing the vertices that are connected to $B_\pm$, respectively, including $B_\pm$ itself.
\end{defn}

See Figure \ref{fig:PositiveNegVer} for an example of how $G_\pm$ are constructed. A subtle point that we want to emphasize is that $G_+$ always contains $B_+$, by definition, but vertices in $\tilde B_+$ do not have to be contained in $G_+$. (More precisely, a vertex in $\tilde B_+$ only lies in $G_+$ if it is connected to $B_+$ via edges not in the good cycle.)

\begin{lm}\label{lm:connected}
The graphs $G_\pm$ are connected graphs. 
\end{lm}

\begin{proof}
Fix an arbitrary vertex $v\in B_+$. By definition, $v$ is the endpoint of a type II edge in $\hat C$; call it $e$. By Lemma \ref{lm:typeIIcycle}, there exists a cycle $\hat C_e$ such that $\hat C_e\cap \hat C=\{\tilde e, e\}$. Since we assume that $B_+$ and $B_-$ are disconnected upon removal of the cycle $\hat{C}$, we conclude that the path $\hat C_e$ connects the vertex $v\in B_+$ to the vertex $\hat v_0\in B_+$ along a path $p$ whose vertices are disjoint from $B_-$.
\end{proof}
 
 Notice that at least one of $B_+$ and $B_-$ contains at least two vertices. Without loss of generality, we assume that $B_+$ contains at least two vertices. Hence, by Lemma \ref{lm:connected}, $G_+$ is a connected graph on at least two vertices and each vertex has even degree. These facts ensure that the notions of preprocessing and good cycle are well-defined for $G_+$.
 
 \begin{lm}
 The graph $G_+$ is not fully preprocessed.
 \end{lm}
 
 \begin{proof}
Assume for a contradiction that $G_+$ is fully preprocessed. Since $G$ is the smallest counterexample and $G_+$ has strictly fewer vertices than $G$, this implies that $G_+$ contains a good cycle $C_{\mathrm{good},+}$, which is good relative to $G_+$. When we embed $C_{\mathrm{good},+}$ into the larger graph $G$, then it must still be a good cycle, a contradiction to the choice of $G$.
 \end{proof}
 
 We now consider the effect of preprocessing the graph $G_+$. 
 Since we assumed that $G$ was the smallest counterexample, $G_{+}$ has a good cycle after undergoing preprocessing. (Notice that the graph $G_+$ has strictly fewer vertices than $G$.)

\begin{defn}
We introduce the sets
$$
\curly{B}_\pm:=B_\pm\cup (\tilde B_\pm\cap G_\pm),
$$
which we call the ``boundary vertices'' of $G_\pm$.
 \end{defn}
 
 We first note that the fact that $B_+$ and $B_-$ are disconnected, also implies that the two corresponding sets of boundary vertices are disconnected in $G\setminus \hat C$. This will allow us to focus on the $+$ case in the following.
 
  \begin{lm}\label{lm:disconnected}
The graphs $\curly{B}_+$ and  $\curly{B}_-$ are disconnected in $G\setminus \hat C$.
 \end{lm}
 
 \begin{proof}
 For a contradiction, suppose that there exists path in $G\setminus \hat C$ connecting two vertices $\tilde b_+\in \curly{B}_+$ and $\tilde b_-\in \curly{B}_-$. The idea is to use this path to construct a path connecting $B_+$ to $B_-$, which will contradict the assumption of case 2.
 
 If $\tilde b_-\in B_-$, then set $b_-=\tilde b_-$. Otherwise, let $b_-\in B_-$ be one of the two nearest vertices to $\tilde b_-$ in the good cycle $\hat C$. Let $f_1\to f_2\to \ldots \to f_r$ be the path in the good cycle connecting $b_-$ and $\tilde b_-$. Note that  $f_1,\ldots,f_r$ are necessarily type I edges; we let $C_{f_i}$ be their associated cycles. We can apply the bypass lemma \ref{lem:bypass} $r$ times to bypass each of the edges $f_i$ with the cycle $C_{f_i}$ (which we note avoids the good cycle and hence lies in $G\setminus \hat C$). The resulting path thus connects $b_-$ to $\tilde b_-$ in $G\setminus \hat C$. The same procedure yields a path from $\tilde b_+$ to a vertex $b_+\in B_+$ in $G\setminus \hat C$ (which is the trivial path if $\tilde b_+\in B_+$ already). We can then use the path that we assumed exists between $\tilde b_-$ and $\tilde b_+$ to construct a path connecting $b_-$ to $b_+$, a contradiction.
 \end{proof}

\begin{prop}\label{prop:prepcharac}
The only possible preprocessing steps that can occur for $G_+$ are the short-circuiting of a degree-2 boundary vertex $b\in \curly{B}_+$ that is connected to two distinct internal vertices in $G_+\setminus \curly{B}_+$.
\end{prop}

The proof of the proposition uses the following lemma which characterizes what incidences can happen at boundary vertices $b\in \curly{B}_+$.

\begin{lm}\label{lm:distance2}
 Any boundary vertex $b\in \curly{B}_+$ cannot satisfy the following:
 \begin{enumerate}[label=(\roman*)]
       \item $b$ cannot be connected to another boundary vertex $b'\in \curly{B}_+$ by a single edge in $G_+$.
     \item Assume additionally that $b$ has degree $2$ in $G_+$. Then $b$ cannot be connected twice to the same internal vertex $w\in G_+\setminus \curly{B}_+$ by a single edge in $G_+$.
 \end{enumerate}
 \end{lm}
 
 \begin{proof}[Proof of Lemma \ref{lm:distance2}]
 Proof of (i). For a contradiction, assume that there is an edge $e_{1,2}$ in $G_+$ between $b_1, b_2 \in \curly{B}_+$. Let $p_{1,2}$ be the part of the cycle $\hat{C}$ between the vertices $b_1$ and $b_2$ that does not contain the edge $\tilde{e}$. We claim that the union $p_{1,2} \cup e_{1,2}$ will be a good cycle in $G$.

First, we consider the edge $e_{1,2}$. Define the cycle $X: =(\hat{C} \setminus p_{1,2}) \cup e_{1,2}$, i.e., the cycle constructed using $e_{1,2}$ and the other part of the cycle $\hat{C}$ not involving $p_{1,2}$. Notice that $X\cap (p_{1,2} \cup e_{1,2})=\{e_{1,2}\}$.

    Next we consider an arbitrary edge $\hat{e}\in p_{1,2}$. There exists a cycle $C_{\hat{e}}$ that will not use any other edge of $p_{1,2}$. Indeed, this is true since we chose $p_{1,2}$ to not use the special edge $\tilde{e}$ and any edge $\hat e$ is either type I or type II. Notice, however, that $C_{\hat{e}}$ might use the edge $e_{1,2}$, a possibility we will now remedy via the bypass lemma. We thus apply the bypass lemma \ref{lem:bypass} to the cycles $X$ and $C_{\hat{e}}$ to construct a new cycle $X_{\hat{e}}$ that uses the edge $\hat{e}$ and satisfies $X_{\hat{e}}\cap (p_{1,2} \cup e_{1,2})=\{\hat{e}\}$. 
    
    This proves that $p_{1,2} \cup e_{1,2}$ is a good cycle in $G$, a contradiction to the choice of $G$. This proves statement (i).
    
    Proof of (ii). Let $b\in \curly{B}_+$ have degree $2$ in $G_+$. For a contradiction, assume that $b$ is connected twice to the same vertex $w$ in $G_{+}\setminus \curly{B}_+$ when restricting to edges in $G_+$. We will now check that the trivial cycle $C_{\mathrm{triv}}$ consisting of these two edges forms a good cycle in $G$ (which then contradicts the choice of $G$). Indeed, by Lemma \ref{lm:connected}, $b$ is connected to $B_+\setminus b$ via edges in $G_+$, and has degree $2$ in $G_+$. Hence, there exists a path $p_+$ in $G_+$ connecting $w$ to some $b'\in B_+\setminus b $. Let $p_{b'\to b}$ be a part of the original good cycle $\hat C$ that connects $b'$ to $b$. Then both edges of $C_{\mathrm{triv}}$ can be composed with $p_+$ followed by $p_{b'\to b}$ to each form a cycle in $G$. These two cycles verify that $C_{\mathrm{triv}}$ is a good cycle. This proves Lemma \ref{lm:distance2}.
 \end{proof}

We are now ready to prove the proposition.

\begin{proof}[Proof of Proposition \ref{prop:prepcharac}]
First, we note that the initial preprocessing step has to occur at one of the boundary vertices $\curly{B}_+$ which has degree less than $2$, since all internal vertices in $G_+\setminus \curly{B}_+$ have degree at least $4$ (because otherwise they could be preprocessed in $G$, which would contradict the choice of $G$). Initially there are no self-loops as $G$ is fully  preprocessed.

Let $b\in \curly{B}_+$ be a boundary vertex where the initial preprocessing step can occur, i.e., $b$ has degree $2$. By Lemma \ref{lm:distance2} (i), $b$ can only be connected to internal vertices in $G_+\setminus \curly{B}_+$ when restricting to edges in $G_+$. By Lemma \ref{lm:distance2} (ii), $b$ is connected to two different vertices $w_1,w_2\in G_+\setminus \curly{B}_+$. Preprocessing then short-circuits $b$, i.e., $b$ is replaced by an additional edge connecting $w_1$ and $w_2$. Afterwards, no further preprocessing steps are necessary at $w_1$ and $w_2$. Indeed, since $w_1$ and $w_2$ are distinct, short-circuiting  cannot create a self-loop and, since $w_1$ and $w_2$ are internal vertices, they have at least degree $4$ (as argued above), so they cannot be subsequently short-circuited. Repeating this procedure for all eligible (i.e., degree-$2$) boundary vertices $b$, we obtain a fully preprocessed graph. This proves Proposition \ref{prop:prepcharac}.
    \end{proof}
    
Now we are ready to give the proof of Proposition \ref{prop:typeIIeven} (and hence Theorem \ref{thm:good}).

\begin{proof}[Proof of Proposition \ref{prop:typeIIeven}]
Assume that there exists an odd number of type II edges. This assumption allows us to define positive and negative vertices and apply the results established in this section. Recall that the goal is now to derive a contradiction. 

Let $P(G_{+})$ be the fully preprocessed graph, where we short-circuit all of the possible boundary vertices as described in Proposition \ref{prop:prepcharac}. Notice that $P(G_{+})$ will necessarily have a good cycle, call it $C_{+}$, by the minimality of $G$. We now check what happens when we elevate this good cycle to the original graph $G$; we will call the lifted cycle $\hat{C}_{+}$. We distinguish the following cases.\\  

\textit{Case (i):} $\hat{C}_{+}$ contains no vertex that gets preprocessed in passing from $G_+$ to $P(G_+)$. Then the lifts of the cycles that establish the fact that $C_+$ is a good cycle in $P(G_+)$ also establish that $\hat C_+$ is a good cycle in $G$.\\

\textit{Case (ii):} $\hat{C}_{+}$ contains exactly one vertex that gets preprocessed in passing from $G_+$ to $P(G_+)$. By Proposition \ref{prop:prepcharac}, this vertex is necessarily a boundary vertex; call it $b\in \curly{B}_+$. Proposition \ref{prop:prepcharac} also implies that $b$ has degree $2$ in $G_+$ and is connected to exactly two vertices $v,w\in G_+\setminus \curly{B}_+$. Recall that the effect of preprocessing the vertex $b$  is to remove it and replace it with an edge $e=(v,w)$.

We claim that $\hat{C}_{+}$ is a good cycle in $G$. First, we consider any edge $e$ in $\hat{C}_{+}$ other than $(b,v)$ and $(b,w)$. The required cycle $C_e$ (intersecting $\hat{C}_{+}$ only at $e$) can be constructed as follows. Recall $C_+$ is a good cycle in $P(G_+)$, so there exists a cycle in $P(G_+)$ that intersects $C_+$ only in $e$. Then we define $C_e$ as the lift of this cycle to $G$, and note that the required condition is verified since the only preprocessing step affecting $\hat C_+$ occurred at $b,v,w$ by assumption.

It remains to find two cycles in $G$ whose intersection with $\hat C_+$ is given by the edge $(b,v)$, respectively $(b,w)$. By applying the bypass lemma \ref{lem:bypass} to the cycles $C_e$ associated to the other edges $e$ in $\hat C_+$ that we constructed above, it is sufficient to find a single cycle containing the edge $(b,v)$, but not $(b,w)$, or vice-versa. To this end, we use Lemma \ref{lm:connected} to find a path $p$ in $G_+$ connecting $b$ to another boundary vertex $b'\in B_+$. Without loss of generality (i.e., by loop erasure), $p$ is simple and thus contains only one of the edges $(b,v)$ or $(b,w)$. We then compose $p$ with a part of the cycle $\hat C$ that connects $b'$ back to $b$. This yields a cycle that contains exactly one of the edges $(b,v)$ or $(b,w)$ as desired. This finishes case (ii) and hence our proof of Proposition \ref{prop:typeIIodd}.\\

\textit{Case (iii):} $\hat{C}_{+}$ contains at least two vertices that get preprocessed in passing from $G_+$ to $P(G_+)$. By Proposition \ref{prop:prepcharac}, these vertices are necessarily boundary vertices in $\curly{B}_+$. We choose two of these, which we call $b_1,b_2\in \curly{B}_{+}$, such that $\hat C_+$ contains a path $p_+$ that connects $b_1$ and $b_2$ and does not visit any other boundary vertex. Moreover, let $p_{1,2}$ be the part of $\hat C_{good}$ that connects $b_1$ with $b_2$ and does not contain $\tilde e$. We claim that the union $p_+\cup p_{1,2}$ is a good cycle in $G$ (which will contradict the choice of $G$). The situation is depicted in Figure \ref{fig:technical}.

\begin{figure}[t]
\begin{center}
    \includegraphics[scale=0.5]{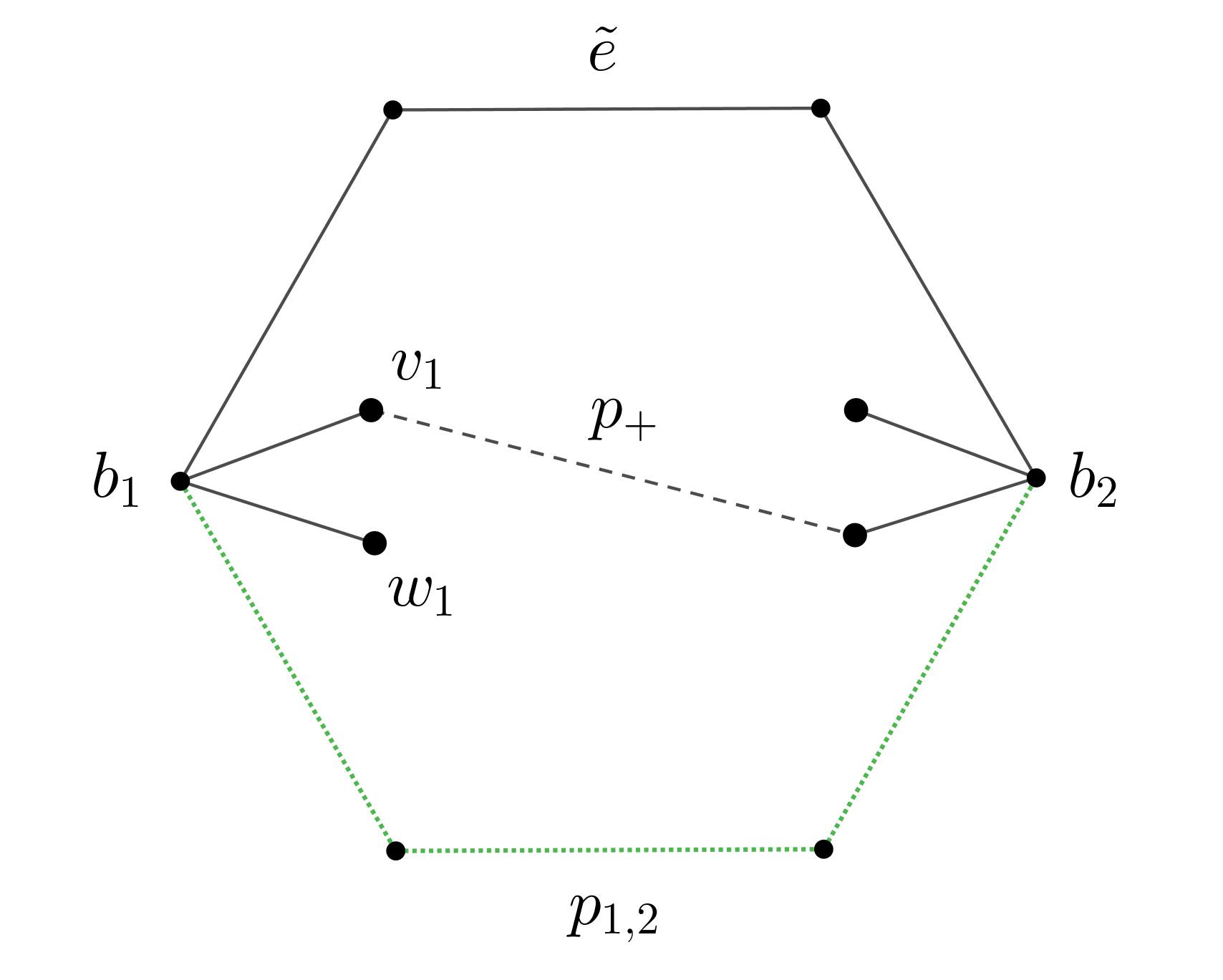}
\end{center}
\caption{A diagram illustrating the construction of the good cycle $p_+\cup p_{1,2}$ in case (iii).}
\label{fig:technical}
\end{figure}

First, consider an edge $e$ of $p_+$. We know that after preprocessing $G_+$, $p_+$ reduces to a part of a good cycle in $P(G_+)$. Hence, there exists a cycle $C_{e,+}$ in $G_+$ such that $C_{e,+}\cap p_+=\{e\}$, and these cycles $C_{e,+}$ embed trivially into $G$ and $C_{e,+}\cap (p_+\cup p_{1,2})=\{e\}$ as desired. (There is a technical point about preprocessing here: According to Proposition \ref{prop:prepcharac}, as, say, $b_1$ gets preprocessed in passing from $G_+$ to $P(G_+)$, it necessarily has degree $2$ in $G_+$ and exactly one of its incident edges is added to build $p_+$ from the good cycle in $P(G_+)$, and we now suppose that the edge $e$ consideration is of this kind. The point is then that also for this kind of edge $e$, there exists the cycle $C_{e,+}$ claimed above. Indeed, since the edge $(v,w)$ in $P(G_+)$ that arises from short-circuiting $b_1$ is necessarily part of the good cycle $C_+$, there exists a cycle $C_{(v,w),+}$ in $P(G_+)$ such that $C_{(v,w),+}\cap C_+=\{(v,w)\}$, and we can take $C_{e,+}$ to be the lift of $C_{(v,w),+}$ to $G_+$. The argument for $b_2$ is analogous. There are no other preprocessing steps necessary for $p_+$ because we chose it such that it does not visit any other boundary vertex in between $b_1$ and $b_2$, and preprocessing of $G_+$ can only occur at the boundary by Proposition \ref{prop:prepcharac}.)

Second, consider an edge $e\in p_{1,2}$. Recall that we write $\hat C_e$ for the cycle that exists either by the definition of type I edge or by Lemma \ref{lm:typeIIcycle}. Since $p_{1,2}$ does not contain $\tilde e$, we have $\hat C_e\cap p_{1,2}=\{e\}$ in either case. However, it is possible that the cycle $\hat C_e$ intersects the path $p_+$ in $G_+$. The solution is to modify $\hat C_e$ using the bypass lemma \ref{lem:bypass} as follows. Suppose that $\hat C_e$ intersects $p_+$ at the edges $e_1,e_2,\ldots,e_K$, which need not be connected. Recall that there exists a cycle $C_{e_1,+}$ in $G_+$ such that $C_{e_1,+}\cap p_+=\{e_1\}$, since $p_+$ is part of a good cycle in $G_+$. Now we apply the bypass lemma \ref{lem:bypass} and use $C_{e_1,+}$ to bypass the edge $e_1$, thereby obtaining a modified cycle $\hat C_e'$ which intersects $p_+\cup p_{1,2}$ at $e_2,\ldots,e_K$. Now we can repeat this procedure to conclude that for every $e\in p_{1,2}$, there also exists a cycle that intersects $p_+\cup p_{1,2}$ only at $e$. This proves that $p_+\cup p_{1,2}$ is a good cycle in $G$, the desired contradiction.
\end{proof}

\section{Quantitative control of subleading graphs}
\label{sect:control}
In this section, we will use our assumption that the frequency sequence $(\om_1,\om_2,\ldots)$ is $(\de,\rho)$-quasi-random with $\de>0$ to control the subleading graphs. Our main result in this section (Proposition \ref{prop:control}) says that the contribution from any graph that is not fully reducible (i.e., that cannot be preprocessed to a point; see Definition \ref{defn:preprocessing}) is subleading in $N$.

We recall that the assumption that $(\om_1,\om_2,\ldots)$ is $(\de,\rho)$-quasi-random means we have the following exponential sum estimate
\beq\label{eq:qrandomestimate}
\l|\frac{1}{N^5} \sum_{i_1,i_2=1}^{\floor{\rho N}} \sum^N_{\substack{j_1,j_2,j_3,j_4=1\\ j_1+j_3=j_2+j_4}} 
e\l[\frac{\om_{i_1}-\om_{i_2}}{2}(j_1^2-j_2^2+j_3^2-j_4^2)\r]\r|\leq N^{-\de}
\eeq

To phrase the main result of this section, we recall the graphical representation of the moment sum in Proposition \ref{prop:graphrep}, i.e.,
$$
\mathbb{E}_{\ul{y}}[\mu^{(2k)}_{M,N}]
=\sum_{G_L=(V,L)\in \curly{L}_k}  \Phi(G_L)
$$
Here we defined
\beq\label{eq:Phidefn}
\Phi(G_L):=\frac{1}{N^{1+k}}\sum_{\substack{\ul{i}=(i_1,\ldots,i_k)\\ L_{\ul{i}}\sim L}}\sum_{\substack{\ul{j}=(j_1,\ldots,j_k)\\ \ul{j}\textnormal{ is $L$-admissible}}} \prod_{r=1}^k K_{(i_r,i_{r+1})}(j_r)
\eeq
with the propagator
$$
K_{(i,i')}(j)= e\l[\frac{\om_i-\om_{i'}}{2}j^2\r].
$$
 The following result establishes that all graph that are not fully reducible are subleading in the moment sum.


\begin{prop}
\label{prop:control}
Assume that \eqref{eq:qrandomestimate} holds. Let $G_L$ be an exploration graph that is not fully reducible. Then
$$
\Phi(G_L)=O(N^{-\de/16}).
$$
\end{prop}

\subsection{The fourth moment case $k=4$}

To clarify the connection between Assumption \eqref{eq:qrandomestimate} and estimates on $\Phi(G_L)$, we begin with the following observation: Assumption \ref{eq:qrandomestimate} is designed to verify Proposition \ref{prop:control} for the fourth moment case $k=4$. Note that for $k=4$ there is only one exploration graph on $4$ edges that is not fully reducible. It is induced by the exploration
\beq\label{eq:Lmelon}
L=((1,2),(2,1),(1,2),(2,1))
\eeq
and we call it the ``melon graph''; see Figure \ref{fig:melon}.
\begin{figure}
    \centering
    \includegraphics[scale=0.1]{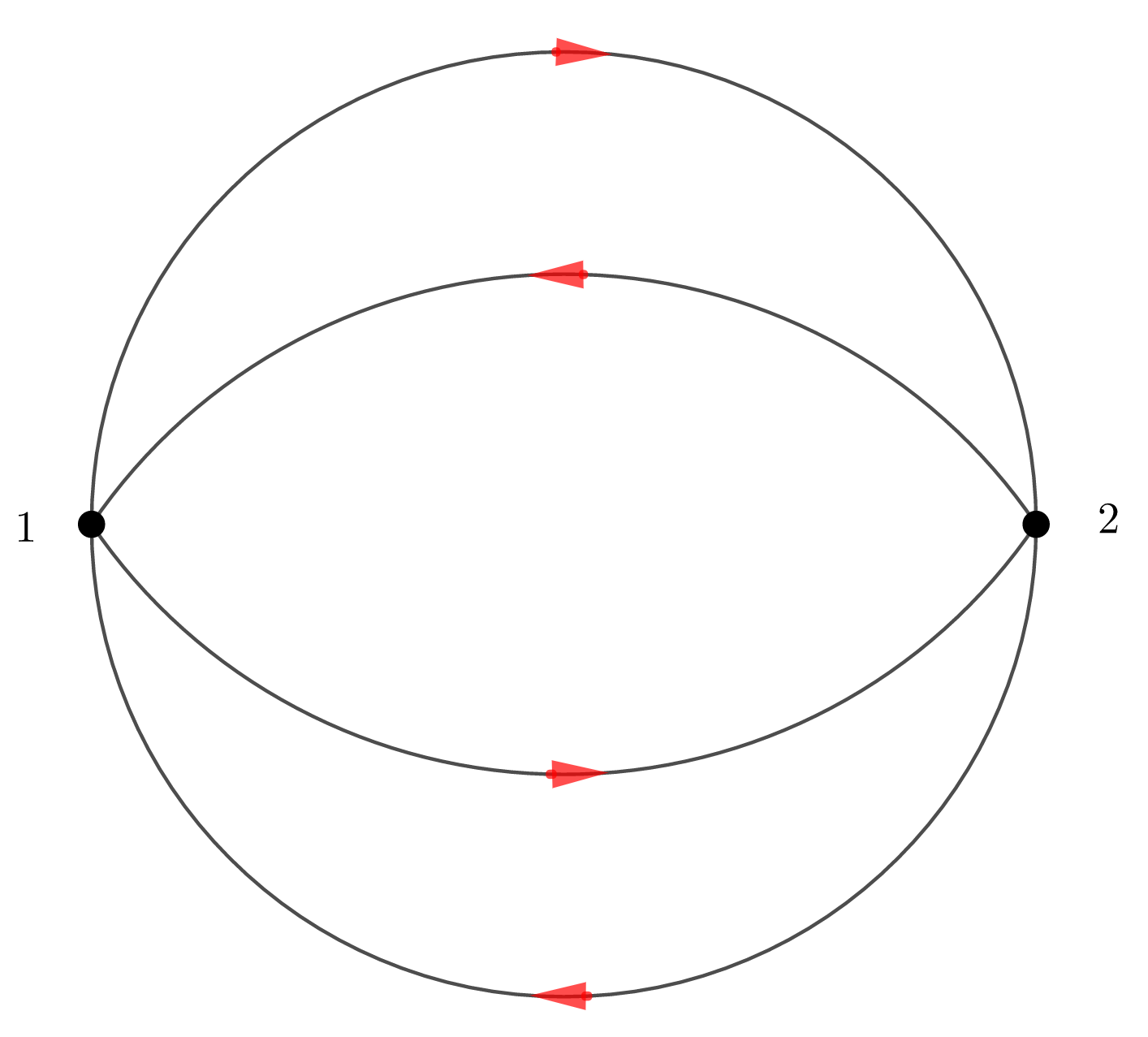}
    \caption{The melon graph defined by $L=((1,2),(2,1),(1,2),(2,1))$.} 
    \label{fig:melon}
\end{figure}

\begin{lm}\label{lem:melon}
The melon graph $G_L$ with $L$ given by \eqref{eq:Lmelon} is subleading, i.e.,
$$
|\Phi(G_L)|\leq N^{-\de}
$$
\end{lm}

\be{proof}
By Definition \eqref{eq:Phidefn} of $\Phi(G_L)$, we have 
$$
\Phi(G_L)=\frac{1}{N^5} \sum_{i_1,i_2=1}^{\floor{\rho N}} \sum^N_{\substack{j_1,j_2,j_3,j_4=1\\ j_1+j_3=j_2+j_4}} 
e\l[\frac{\om_{i_1}-\om_{i_2}}{2}(j_1^2-j_2^2+j_3^2-j_4^2)\r].
$$
The bound thus holds by Assumption \eqref{eq:qrandomestimate}.
\end{proof}

In the remainder of this section, we show that control of the $k=4$ term via Assumption \eqref{eq:qrandomestimate} is in fact sufficient to control the subleading graphs for all moments. A crucial input to the result is the existence of a good cycle in any subleading graph, i.e., Theorem \ref{thm:good}. 

\subsection{Invariance of the leading term under preprocessing}
In order to use the graph-theoretical results from the previous section, we first establish that preprocessing a graph affects $\Phi(G_L)$ in a simple way, up to error terms.

\begin{lm}\label{lm:preproinv}
Let $G_L\in \curly{L}_k$ be an exploration graph on $k$ edges and $l$ vertices. Then, we have
$$
\begin{aligned}
\Phi(G_L)=&\rho\Phi(\curly{S}(G_L,v))+O(N^{-1})\\ \Phi(G_L)=&\Phi(\curly{L}(G_L,l)).
\end{aligned}
$$
for any vertex $v\in G_L$ that can be short-circuited and any self-loop $l$ in $G_L$.
\end{lm}

\begin{proof}
\textit{Case 1: Short-circuiting.} Let $v\in G_L$ be a vertex that can be short-circuited. For each $L_{\ul{i}}\sim L$, there is a unique $1\leq r\leq k$ such that $i_r=v$. By Kirchhoff's law, the current incoming to $v$ equals the current outgoing from $v$, i.e., $j_{r}=j_{r-1}$. Hence
\beq\label{eq:sccancel}
\begin{aligned}
K_{(i_{r-1},i_r)}(j_{r-1})K_{(i_r,i_{r+1})}(j_r)
=&e\l[(\om_{i_{r-1}}-\om_{i_{r}})\frac{j_{r}^2}{2}\r]
e\l[(\om_{i_{r}}-\om_{i_{r+1}})\frac{j_{r}^2}{2}\r]\\
=&e\l[(\om_{i_{r-1}}-\om_{i_{r+1}})\frac{j_{r}^2}{2}\r]\\
=& K_{i_{r-1},i_{r+1}}(j_r).
\end{aligned}
\eeq
We conclude that short-circuiting provides an $(M-l+1)$-fold mapping 
$$
\{(i_1,\ldots, i_k),(j_1,\ldots, j_k)\}
\to \{(i_1,\ldots,\hat i_r,\ldots i_k),(j_1,\ldots,\hat j_r,\ldots j_k)\}
$$
(where $\hat i_r$, $\hat j_r$ means that those indices are skipped), and the mapping preserves the admissibility (Kirchhoff current law) and the associated propagator product. Recall that $M=\floor{\rho N}$. By \eqref{eq:sccancel}, we conclude that
$$
\Phi(\curly{S}(G_L,v))=\frac{M-l+1}{N}\Phi(G_L)
=\rho \Phi(G_L)+\frac{-l+1+\floor{\rho N}-\rho N}{N}\Phi(G_L).
$$
To control the error term, we notice that we have the a priori bound
$$
\begin{aligned}
|\Phi(G_L)|
\leq& \frac{1}{N^{1+k}}\sum_{\substack{\ul{i}=(i_1,\ldots,i_k)\\ L_{\ul{i}}\sim L}}\sum_{\substack{\ul{j}=(j_1,\ldots,j_k)\\ \ul{j}\textnormal{ is $L$-admissible}}}1\\
=&\frac{1}{N^{1+k}}\frac{M!}{(M-l)!}
\sum_{\substack{\ul{j}=(j_1,\ldots,j_k)\\ \ul{j}\textnormal{ is $L$-admissible}}}1\\
\leq &C_k\rho^{l}. 
\end{aligned}
$$
In the last step, we used Corollary \ref{cor:K} on the number of $L$-admissible current assignments. We have thus shown that
$$
\Phi(\curly{S}(G_L,v))=\rho\Phi(G_L)+O(N^{-1})
$$
as claimed.\\

\textit{Case 2: Loop removal.} Let $l$ be a self-loop in $G_L$.
This means there exists an $1\leq r\leq k$ such that $i_{r}=i_{r+1}$ with associated current $j_r$. Notice that 
$$
K_{(i_{r},i_{r+1})}(j_r)= e\l[(\om_{i_r}-\om_{i_r})\frac{j_r^2}{2}\r]=1.
$$
We conclude that short-circuiting provides an $N$-fold mapping
$$
\ul{j}=(j_1,\ldots,j_k)\to (j_1,\ldots,\hat{j}_r,\ldots,j_k)
$$
and the mapping preserves the admissibility (Kirchhoff current law). Hence,
$$
\Phi(\curly{L}(G_L,l))=\Phi(G_L).
$$
This finishes the proof of Lemma \ref{lm:preproinv}.
\end{proof}

\subsection{Proof of Proposition \ref{prop:control}}

Thanks to Lemma \ref{lm:preproinv}, we may assume from now on that all graphs in the moment computations are fully preprocessed.

Let us summarize the proof strategy for Proposition \ref{prop:control}. The proof rests on two key Lemmas \ref{lem:PhitoESB} and \ref{lem:ESBtoPhi}. On the one hand, Lemma \ref{lem:PhitoESB} assumes that one knows $\Phi(\curly{C}_n)$ is subleading, where $\curly{C}_n$ is the doubly-traversed cycle on $n$ vertices,
$$
\curly{C}_n=((1,2),(2,3),\ldots,(n-1,n),(n,1),(1,n),(n,n-1),\ldots,(2,1)),
$$
and from this it derives a certain bound on exponential sums for most frequencies (via the pigeonhole principle). This exponential sum bound is called $ESB_n(\kappa)$ below. 

On the other hand, Lemma \ref{lem:ESBtoPhi} assumes $ESB_n(\kappa)$ and derives a bound on $\Phi(G_L)$ for any graph containing a good cycle of $n$ vertices. The key idea here is that the estimate $ESB_n$ is designed to exploit the oscillations in the exponential sum associated to the net current that runs through the good cycle; we call this current $A$ below.\\

We begin with preliminaries for phrasing the key lemmas mentioned above. As explained above imperative to reparametrize the current assignments $\ul{j}$ in \eqref{eq:Phidefn} satisfying the Kirchhoff current law such that $A$ is its own variable of summation. Notice that the exponential sum $ES$ define in \eqref{eq:ESdefn} below only involves oscillations from $A$. The set $T(\ul{t},\ul{\sigma})$ defined in \eqref{eq:Tdefn} below describes the constraints imposed on $A$ from the remaining graph data.

Throughout the argument, we write $C_k$ for a positive constant that may depend on $k$ and $\rho$, but not on $N$.

Given a vector of integers $\ul{t}=(t_1,\ldots, t_{n-1},0)\in\Z^{n-1}\times\{0\}$ and signs $\ul{\sigma}=(\sigma_1,\ldots,\sigma_n)\in \{\pm 1\}^n$, we define the set 
\beq\label{eq:Tdefn}
T(\ul{t},\ul{\sigma})
:=\setof{A\in \Z}{1\leq t_r+A\sigma_r\leq N,\quad\forall 1\leq r\leq n}.
\eeq
and the exponential sum
\beq\label{eq:ESdefn}
ES(\om_{i_1},\ldots,\om_{i_n},\ul{t},\ul{\sigma}):=
\l|\sum_{A\in T(\ul{t},\ul{\sigma})} e\l[\sum_{r=1}^n\sigma_r(t_r+A\sigma_r)^2\frac{\om_{i_{r+1}}-\om_{i_r}}{2}\r]\r|
\eeq

\begin{defn}
Let $n\geq 1$ and $\kappa>0$. We say that the statement $ESB_n(\kappa)$ holds, if, for all choices of signs $\ul{\sigma}\in \{\pm 1\}^n$ the ``bad set''
$$
\setof{(i_1,\ldots,i_n,\ul{t})\in \{1,\ldots,M\}^n\times \Z^{n-1}\times\{0\}}{ES(\om_{i_1},\ldots,\om_{i_n},\ul{t},\ul{\sigma})>N^{1-\kappa}}
$$
has cardinality bounded by $C_k N^{2n-1 - \kappa}$.
\end{defn}

The statement $ESB_n(\kappa)$ says that oscillations make the key exponential sum $ES$, which will be associated to the current along the good cycle, $\ll N$ for most frequencies and external currents.

\begin{lm}\label{lem:PhitoESB}
Let $n\geq 1$ and $\kappa>0$. If $|\Phi(\curly{C}_n)|\leq C_k N^{-\kappa}$, then $ESB_n(\kappa/4)$.
\end{lm}

\begin{lm}\label{lem:ESBtoPhi}
Suppose that $ESB_n(\kappa)$ holds. Let $L$ be an exploration on $k$ edges and $l$ vertices that is fully preprocessed and has a good cycle of length $n$. Then
$$
|\Phi(G_L)|\leq C_k N^{-\kappa}.
$$
\end{lm}

We postpone the proof of these lemmas for now. 

\begin{proof}[Proof of Proposition \ref{prop:control}]
To start the proof, notice that the melon
$$
L=((1,2),(2,3),(3,4),(4,1))=\curly{C}_2,
$$
is the doubly traversed cycle on $2$ vertices. Hence, Lemma \ref{lem:melon} establishes
$$
|\Phi(\curly{C}_2)|\leq N^{-\de}.
$$
Now we apply Lemma \ref{lem:PhitoESB} with $\kappa=\de$ and obtain
$ESB_2(\de/4)$. Let $m\leq k$ be an arbitrary integer. Observe that $\curly{C}_m$ has a good cycle of length $2$ (in fact, many of them). Hence, we can apply Lemma \ref{lem:ESBtoPhi} with $n=2$ to find
$$
|\Phi(\curly{C}_m)|\leq C_kN^{-\de/4}.
$$
By Lemma \ref{lem:PhitoESB}, we obtain the exponential sum bound $ESB_m\l(\frac{\de}{16}\r)$ for any integer $1\leq m\leq k$.

Finally, let $G_L$ be an arbitrary exploration graph in $\curly{L}_k$ that is not fully reducible. By Theorem \ref{thm:good}, $G_L$ contains a good cycle. Let $m$ be the length of a good cycle in $G_L$. By the exponential sum bound $ESB_m\l(\frac{\de}{16}\r)$ established above and Lemma \ref{lem:ESBtoPhi}, we conclude
$$
|\Phi(G_L)|\leq C_k N^{- \frac{\de}{16}}
$$
This proves Proposition \ref{prop:control}.
\end{proof}

\be{rmk}
The formula for $\Phi(\curly{C}_m)$ contains a term which is proportional $N^{-1}$, which indicates that the decay rate of $\Phi(G_L)$ does not generally improve with the length of the good cycle. 
\e{rmk}

\subsection{Proofs of Lemma \ref{lem:PhitoESB}}
We start from the cycle $\curly{C}_n$ and a choice of $\ul{\sigma}=(\sigma_1,\ldots,\sigma_n)\in{\pm 1}$. We then decompose $\curly{C}_n$ into two subcycles $\curly{C}_{\ul{\sigma}}=(e_1,\ldots,e_n)$ and $\tilde{\curly{C}}_{\ul{\sigma}}=(\tilde e_1,\ldots,\tilde e_n)$ as follows:
$$
\begin{aligned}
\sigma_l=1\,&\Longleftrightarrow\, e_l:=(l,l+1),\textnormal{ and }\, \tilde e_l:=(l+1,l), \\
\sigma_l=-1\,&\Longleftrightarrow\, e_l:=(l+1,l),\textnormal{ and }\, \tilde e_l:=(l,l+1),
\end{aligned}
$$
with the convention that $n+1=1$. 

Recall that $\Phi(\curly{C}_n)$ involves a choice of currents $\ul{j}=(j_{e_1}\ldots,j_{e_n},j_{\tilde e_1}\ldots,j_{\tilde e_n})\in \{1,\ldots,N\}^{2n}$ satisfying the Kirchhoff law. We parametrize these by
$$
\begin{aligned}
j_{e_r}=&t_r+A\sigma_r,\qquad \forall 1\leq r\leq n-1,\\
j_{\tilde e_r}=&t_r+B\sigma_r,\qquad \forall 1\leq r\leq n-1\\
j_{e_n}=&A\sigma_n,\qquad j_{\tilde e_n}=B\sigma_n.
\end{aligned}
$$
The variable $t_r$ represents the net current running around the short cycle $e_r\to \tilde e_{r}\to e_r$, while $A$ (respectively $B$) is the net current running around $C_{\ul{\sig}}$ (respectively $\tilde{C}_{\ul{\sig}}$). Implementing this parametrization, we see that
$$
\Phi(\curly{C}_n)=\frac{1}{N^{1+2n}}\sum_{i_1,\ldots,i_n=1}^{\floor{\rho N}}\sum_{t_1,\ldots,t_{n-1}\in \Z}
ES(\om_{i_1},\ldots,\om_{i_n},\ul{t},\ul{\sigma})^2.
$$
Here we used that $A\in T(\ul{t},\ul{\sigma})\Leftrightarrow B\in T(\ul{t},\ul{\sigma})$ to complete a square. By Assumption, $|\Phi(\curly{C}_n)|\leq C_n N^{-\kappa}.$ Notice that the constant $C_n$ is independent of the choice of signs $\ul{\sigma}=(\sigma_1,\ldots,\sigma_n)\in\{\pm 1\}^n$. The pigeonhole principle then implies $ESB_n(\kappa/4)$.
\qed

\subsection{Proof of Lemma \ref{lem:ESBtoPhi}}
Let $C_g=(e_1,\ldots,e_{n})$ be a good cycle in $G_L$ of length $n$. For each $1\leq i\leq n$, let $C_{e_i}$ be a cycle such that $C_{e_i}\cap C_g=\{e_i\}$; these cycles exist by Definition \ref{defn:good} of a good cycle. To compute $\Phi(G_L)$, we need to select currents $\ul{j}\in \{1,\ldots,N\}^k$. By Corollary \ref{cor:K}, there are effectively $k-l+1$ free current variables. Following Appendix A, we first consider the set of solutions to Kirchhoff's law as an $\R$-vector space. (For now, we ignore the constraint that the current along each edge should be in $\{1,\ldots,N\}$.) It is a subspace of $\R^k$, and each vector component represent an assignment of $\R$-valued current through a corresponding edge. Using Definition \ref{defn:curlyc}, we first select the linearly independent collection of our internal cycles
$$
\{\de_{C_g},\de_{C_{e_1}},\ldots,\de_{C_{e_{n-1}}}\}\subset \tilde{\curly{C}},
$$
where the orientation of $\de_C$ is inherited from the orientation of the edges of each cycle $C$ as described in Definition \ref{defn:curlyc}. Next, we extend this set to a basis of $\tilde{\curly{C}}$ by adding a basis $\tilde{\curly{B}}$ associated to all the remaining (``external'') cycles. By Theorem \ref{thm:K}, $\tilde{\curly{B}}$ has $m=k-l+1-n$ elements. Finally, we use the isomorphism between $\tilde{\curly{C}}$ and the $\R$-vector space of solutions of Kirchhoff's current law, $\curly{C}$, that was established in Lemma \ref{lm:iso}. This isomorphism maps our basis $\tilde{\curly{C}}$ to a basis of current assignments that automatically satisfy Kirchhoff's current law with $k-l+1$ parameters:
$$
\curly{B}\cup \{b_g,b_{e_1},\ldots,b_{e_{n-1}}\}\subset \curly{C}\subset \R^k,
\qquad \curly{B}=\{b_1,\ldots,b_{m}\}.
$$

Since we want to extract cancellations from the current along the good cycle, the idea is to use Fubini to take the sum over the external currents outside. We call the collection of external currents $j_{\curly{B}}$. 

Having constructed the bases, we now have to implement the constraint that the net current crossing each edge should be $\in\{1,\ldots,N\}$ by formula \eqref{eq:Phidefn} for $\Phi(G_L)$.

A choice of external currents $j_{\curly{B}}$ is admissible if and only if there is a corresponding choice of internal currents such that collectively, they satisfy that constraint. (Note that they automatically satisfy the Kirchhoff current law by design.) We recall that, by definition, the $k$ components of a vector in $\curly{C}$ describe the current assigned to the $k$ edges of the graph. Hence, we have the formal constraint
\beq\label{eq:JBdefn}
\begin{aligned}
j_\curly{B}\in \curly{J}_{\curly{B}}:=&\bigg\{(a_1,\ldots,a_m)\in \R^m\,:\,
\exists (a_g,a_{e_1},\ldots,a_{e_{n-1}})\in \R^n\,:\,\\
&\quad\sum_{i=1}^m a_i b_i+a_g b_g+a_{e_1}b_{e_1}+\ldots+a_{e_{n-1}}b_{e_{n-1}}\in \{1,\ldots,N\}^{k-l+1}\bigg\}.
\end{aligned}
\eeq
Given a $j_\curly{B}\in \curly{J}_{\curly{B}}$, we can assign the $n$ currents for the internal cycles. We parametrize these similarly as in the proof of Lemma \ref{lem:PhitoESB}, i.e., we set
$$
\begin{aligned}
j_{e_r}=&t_r+A\sigma_r,\qquad \forall 1\leq r\leq n-1,\\
j_{e_n}=&A\sigma_n,
\end{aligned}
$$
where $\sigma_r$ encodes the orientation of the edge $e_r$. See Figure \ref{fig:internalexternal} for an example.

\begin{figure}[t]
    \centering
    \includegraphics[scale=0.2]{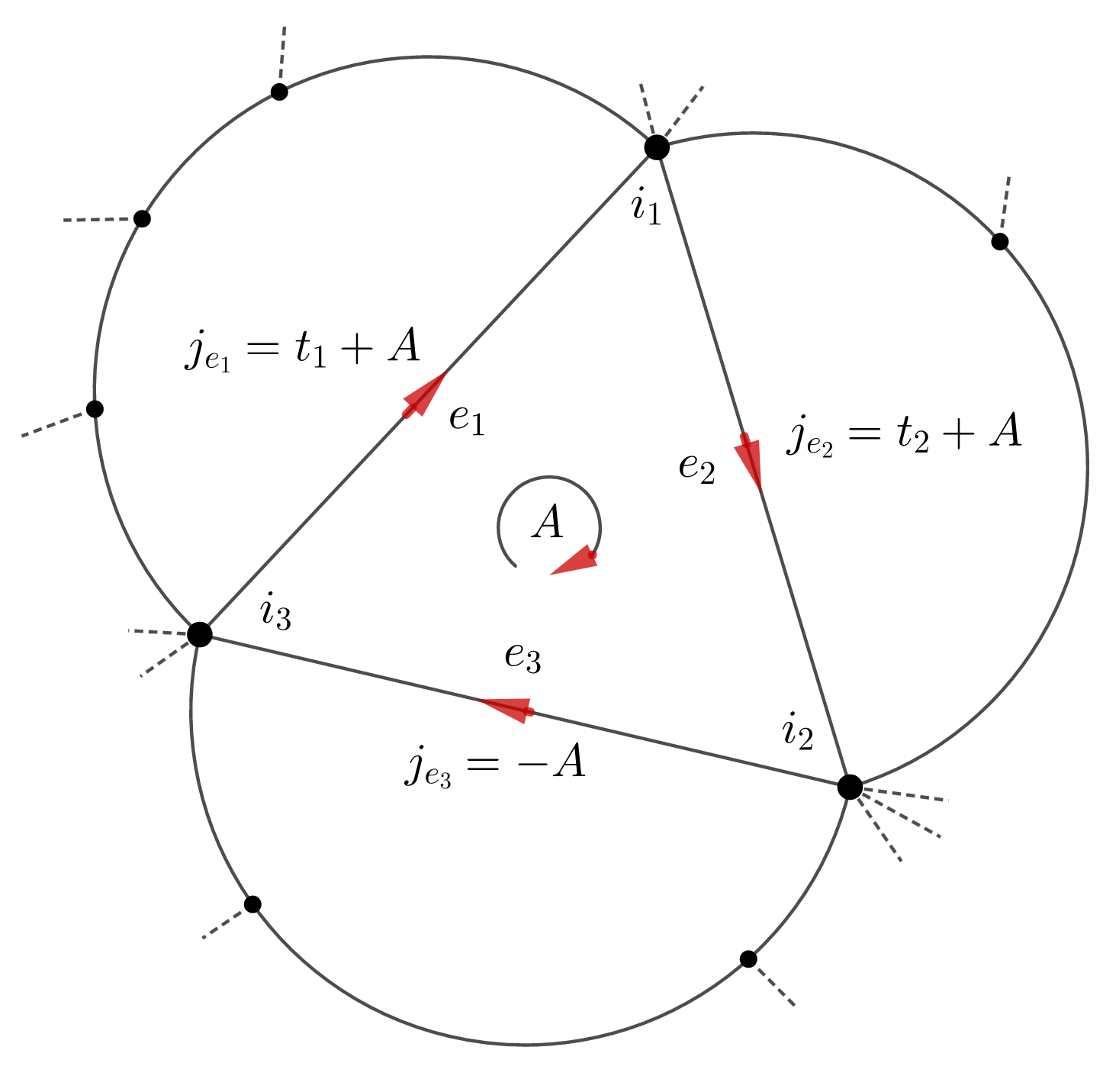}
    \caption{An example for the set of internal edges and vertices; the dashed lines connect to the external edges and vertices. The good cycle is $C_g=e_1\to e_2\to e_3\to e_1$, and the semicircles represent the cycles $C_{e_i}$ which intersect $C_g$ only at $e_i$. The key variable from which we extract cancellations is the net current $A$ running through the good cycle $C_g$.}
    \label{fig:internalexternal}
\end{figure}

The idea behind this parametrization is that the external currents $j_\curly{B}$ can be seen as effectively injecting a certain outside current into each edge $e_r$. Modulo a shift that removes this injected current, we can think of $t_r$ as the net current running along the cycle $C_{e_r}$ as before. Similarly, modulo an injection at edge $e_n$, $A$ represents the net current running along the cycle $C_g$. Since the shift depends on $j_\curly{B}$, so does the set of admissible choices of $\ul{t}=(t_1,\ldots,t_{n-1},0)$, and we call this set $\curly{T}(j_\curly{B})$. The current $A$ is chosen last and is constrained in the same way as in Lemma \ref{lem:PhitoESB}, i.e., $A\in T(\ul{t},\ul{\sigma})$.

Finally, we decompose the vertices into external and internal ones in a similar fashion. We write $\{i_1,\ldots,i_n\}$ for the vertices visited by $C_g$ (``internal'' vertices) and $\curly{I}$ for the remaining $l-n$ vertices in the graph $G_L$ (``external'' vertices). Moreover, we use the shorthand for the external propagator
$$
K_{\curly{I}}(j_{\curly{B}})
$$
which is the product of propagators $K_{(i,i')}(j_e)$ for all choices of external edges $e$ (so at least one of the vertices $i,i'$ is external as well).

By Fubini, we can now move the sum over the external variables outside, and the sum over $A$ inside, and express $\Phi(G_L)$ as 
$$
\begin{aligned}
&\Phi(G_L)=\frac{1}{N^{k+1}}\\
&\times
\underbrace{\sum_{\substack{\curly{I}\subset\{1,\ldots,M\}:\\ |\curly{I}|=l-n}}\sum_{j_\curly{B}\in \curly{J}_{\curly{B}}} K_{\curly{I}}(j_{\curly{B}})}_{\textnormal{external sum}}
\underbrace{\sum'_{\substack{i_1,\ldots,i_n\\ \in \{1,\ldots,M\}\setminus \curly{I}}}\sum_{\ul{t}\in \curly{T}(j_\curly{B})}\sum_{A\in T(\ul{t},\ul{\sigma})} \prod_{r=1}^n
K_{(i_r,i_{r+1})}(t_r+A\sigma_r)}_{\textnormal{internal sum}}.
\end{aligned}
$$

Here $\sum'$ means that the sum is taken over distinct indices $i_1,\ldots,i_n$. By the triangle inequality and the fact that all external propagators are bounded in modulus by $1$, we obtain
\begin{align}
|\Phi(G_L)|\leq& \frac{1}{N^{k+1}}
\sum_{\substack{\curly{I}\subset\{1,\ldots,M\}:\\ |\curly{I}|=l-n}}
\sum_{j_\curly{B}\in \curly{J}_{\curly{B}}}\sum'_{\substack{i_1,\ldots,i_n\\ \in \{1,\ldots,M\}\setminus \curly{I}}}\sum_{\ul{t}\in \curly{T}(j_\curly{B})}\l|\sum_{A\in T(\ul{t},\ul{\sigma})} \prod_{r=1}^n
K_{(i_r,i_{r+1})}(t_r+A\sigma_r)\r|\\
=& \frac{1}{N^{k+1}} \sum_{j_\curly{B}\in \curly{J}_{\curly{B}}}\sum_{\ul{t}\in \curly{T}(j_\curly{B})}
\sum'_{\substack{i_1,\ldots,i_n\\ \in \{1,\ldots,M\}}}  \binom{M-n}{l-n}\l|\sum_{A\in T(\ul{t},\ul{\sigma})} \prod_{r=1}^n
K_{(i_r,i_{r+1})}(t_r+A\sigma_r)\r|
\end{align}
In the second step, we first merged the sums over the vertices and then performed the sum over the external vertices given the internal ones. Now we recall our assumption $ESB_n(\kappa)$ which says that
$$
\l|\sum_{A\in T(\ul{t},\ul{\sigma})} \prod_{r=1}^n
K_{(i_r,i_{r+1})}(t_r+A\sigma_r)\r|\leq N^{1-\kappa}
$$
for all collections of internal vertices and currents $(i_1,\ldots,i_n,\ul{t})\in \{1,\ldots,M\}^n\times T(\ul{t},\ul{\sigma})$, excluding a bad set of cardinality bounded by $C_kN^{2n-1-\kappa}$. On the bad set, we can use the trivial bound 
$$
\l|\sum_{A\in T(\ul{t},\ul{\sigma})} \prod_{r=1}^n
K_{(i_r,i_{r+1})}(t_r+A\sigma_r)\r|\leq |T(\ul{t},\ul{\sigma})|\leq N
$$
Decomposing the sums over the good and bad sets, we obtain
\beq\label{eq:Phiestimate}
\begin{aligned}
|\Phi(G_L)|\leq& \frac{1}{N^{k+1}}\binom{M-n}{l-n}\sum_{j_\curly{B}\in \curly{J}_{\curly{B}}}
\l(\sum_{\ul{t}\in \curly{T}(j_\curly{B})} \sum'_{\substack{i_1,\ldots,i_n\\ \in \{1,\ldots,M\}}} N^{1-\kappa}+C_k N^{2n-\kappa}\r)\\
=&\frac{1}{N^{k+1}}\binom{M-n}{l-n}\l(\binom{M}{n} N^{1-\kappa} \sum_{j_\curly{B}\in \curly{J}_{\curly{B}}}
|\curly{T}(j_\curly{B})| 
+  |\curly{J}_{\curly{B}}| C_k N^{2n-\kappa}\r)
\end{aligned}
\eeq
It remains to control the cardinalities $|\curly{T}(j_\curly{B})|$ and $|\curly{J}_{\curly{B}}|$ which count the number of admissible assignments of external currents. For this we use Lemma \ref{lem:latticepoints2}. Indeed, since Definition \eqref{eq:JBdefn} is exactly of the required form \eqref{eq:latticepointform} with $m'=k-l+1$ and $m=k-l+1-n$ as before, Lemma \ref{lem:latticepoints2} gives the bound
$$
|\curly{J}_{\curly{B}}|\leq C_{m'} N^{m}\leq  C_k N^{k-l+1-n}.
$$
Similarly, we can bound $|\curly{T}(j_\curly{B})|$ by using Lemma \ref{lem:latticepoints2} with the choices $m'=n$, $m=n-1$ and shift vector determined by $j_{\curly{B}}$. This gives the bound
$$
|\curly{T}(j_\curly{B})|\leq C_{n} N^{n-1}\leq C_k N^{n-1},
$$
Here and above, we took $C_k$ to be the maximum over the finitely many constants $C_n$, with $1\leq n\leq k$. We apply these cardinality bounds to \eqref{eq:Phiestimate} and use
$$
\binom{M}{n}\leq C_k M^n=C_kN^n,\qquad \binom{M-n}{l-n}\leq C_k M^{l-n}= C_k N^{l-n}.
$$
(Recall that $M=\floor{\rho N}$.) This yields
$$
\begin{aligned}
|\Phi(G_L)|\leq C_k\frac{N^{l-n}}{N^{k+1}}\l(N^n N^{1-\kappa} N^{k-l+1-n}
N^{n-1}
+   N^{k-l+1-n} N^{2n-\kappa}\r)= C_k N^{-\kappa}.
\end{aligned}
$$
This proves Lemma \ref{lem:ESBtoPhi}.
\qed

\section{Identifying the limiting distribution}
In this section, we will count the contribution to the moments of leading graphs. We will do this by deriving a recursion for the weight of graphs that are leading.

In this section, we consider exploration graphs as in Definition \ref{defn:exploration}. We note that by forgetting the direction of edges, any exploration $L$ can be viewed as a graph $G_L=(V,L)$. In this way, we can define preprocessing on exploration graphs and a fully preprocessed exploration graph $G_L$ is one to which no further preprocessing steps can be applied. 

\begin{defn}[Fully reducible exploration graphs]
Let $G_L=(V,L)$ be an exploration graph. We say $G_L$ is ``fully reducible'', if $G_L$ can be preprocessed to a point.
\end{defn}

We need the following lemma on fully reducible exploration graphs.

\begin{lm} \label{l:edgdisjoint}
Let $G_L=(V,L)$ be an exploration graph. 
If there exist two vertices $v_1,v_2 \in V$ such that there are 4 edge-disjoint paths $E_1$, $E_2$, $E_3$, $E_4$ in between them, then $L$ is not fully reducible. 
\end{lm}
\begin{proof}
Consider two vertices $v_1^0$ and $v_2^0$ that have four edge-disjoint paths called $E_1^0, E_2^0, E_3^0, E_4^0$ in between them. We will show by induction that after each preprocessing step, which will be indexed by time $t$, we can still find two vertices $v_1^t$ and $v_2^t$ such that there exist four edge disjoint paths $E_1^t$, $E_2^t$, $E_3^t$ and $E_4^t$ in between $v_1^t$ and $v_2^t$ in the $t$ step preprocessed graph $G^t$.

If we do not short-circuit a vertex that is in one of the paths $E_1^t$, $E_2^t$, $E_3^t$ or $E_4^t$, then we may set $E_j^{t+1} = E_j^{t}$ and $v_{1,2}^{t+1} = v_{1,2}^t $. Similarly, removing a loop cannot affect the existence of the edge disjoint paths.

The remaining possibility is that short-circuiting removed a vertex inside one of the $E_i^t$. Consider the case that said vertex is not $v_1^t$ or $v_2^t$. Then this vertex cannot belong to two of the edge disjoint paths $E_{j_1,j_2}^t$; a vertex belong to two such edge disjoint paths must have degree at least four. Thus, it would not be chosen for preprocessing. Said vertex would belong to exactly one of the $E_j^t$; in this case short-circuiting would reduce the size of the path, but not completely remove it. The only remaining case is that preprocessing short-circuits either $v_1^t$ or $v_2^t$. However, for this to be possible, either $v_1^t$ or $v_2^t$ has degree $2$. This is clearly not possible when there are four edge disjoint paths in between $v_1^t$ and $v_2^t$. 

This shows that at all times $t$, we have two points $v_1^t$ and $v_2^t$ such that there are $4$ edge disjoint paths in between $v_1^t$ and $v_2^t$. 
This graph cannot be preprocessed into a point.
\end{proof}

\begin{defn}[Total weights of reducible graphs]
Let $\curly{R}_k$ denote the set of fully reducible exploration graphs with $k$ edges. Given an exploration graph $G_L=(V,L)$ in $\curly{R}_k$, we define the weight function $W[G_L] = \rho^{|V|}$ where $\rho$ is the constant from Theorem \ref{thm:keyrho}. We define
\begin{equation}
\tilde\mu_k:= \sum_{G_L\in\curly{R}_k} W[G_L]
\end{equation}
Notice that $\tilde\mu_0 =1$.
\end{defn}

\begin{thm}[Recursion for the limiting moments]
\label{thm:recursion}
The $(\tilde\mu_k)_{k\geq 1}$ satisfy the following recurrence relation.
\begin{equation}
\label{eq:recursion}
    \tilde\mu_k =  \tilde\mu_{k-1} + \sum_{n=2}^{k}  \rho \tilde\mu_{n-1} \tilde\mu_{k- n},\qquad \forall k\geq 2,
\end{equation}
with initial values $\tilde\mu_0 =1$ and $\tilde\mu_1=\rho$. 
\end{thm}

\begin{proof}[Proof of Theorem \ref{thm:recursion}]
We first need to find a natural way to decompose a reducible graph into component parts which we are also able to recognize as reducible.

We start with a reducible graph $G$ with initial starting vertex $v$. We will define the subgraph $C_1$ to be the subgraph of $G$ in between $v$ and the first return to $v$. The subgraph $C_2$ will be the subgraph of $v$ between the first return to $v$ and the final return to $v$, or alternatively all the way to the end. We claim that the graphs $C_1$ and $C_2$ are vertex-disjoint aside from the initial vertex $v$.

Assume for contradiction that there is another vertex $w$ in common between $C_2$ and $C_1$. Consider $C_1^{v \rightarrow w}$ to be the part of $C_1$ traversed starting from $v$ to the first appearance of $w$ and let $C_1^{w \rightarrow v}$ be the part of $C_1$ from the first appearance of $w$ to the end of $C_1$ when it returns to $v$. Define $C_2^{v \rightarrow w}$ and $C_2^{w \rightarrow v}$ accordingly. Then $C_1^{w \rightarrow v}$, $C_1^{v \rightarrow w}$, $C_2^{w \rightarrow v}$ and $C_2^{v \rightarrow w}$ form 4 edge disjoint paths in between the vertices $v$ and $w$. This contradicts lemma \ref{l:edgdisjoint} since we assumed $G$ is fully reducible. Thus $C_1$ and $C_2$ are vertex-disjoint graphs excluding $v$. Additionally, this shows that the subgraphs $C_1$ and $C_2$ are themselves reducible graphs.

Similar to how one proves the Catalan recurrence, we see that one now has a product structure for reducible graphs. Namely, we see that $\tilde\mu_k$ should be the sum over all $n$ of the product of total weight of all possibilities for $C_1$ when $C_1$ has $n$ edges and the total weight of all possibilities of $C_2$ when $C_2$ has $k-n$ edges. The total weight for all possibilities for $C_2$ when $C_2$ has $k-n$ edges is merely $\tilde\mu_{k-n}$ as $C_2$ can vary over all all reducible graphs of size $k-n$. There is slightly more difficulty in counting the total weight of all possibilities of $C_1$ since we cannot return to the first vertex $v$ until the end. However, notice that this implies that the vertex $v$ can be short-circuited. This implies that the graphs forming $C_1$ are in bijection with all reducible graphs with $n-1$ edges. 

If $n$ is greater than $1$, this means that preprocessing removes a vertex and the total weight of all graphs forming $C_1$ is $\rho \tilde\mu_{n-1}$. If $n=1$, then preprocessing removes a loop and we do not need to add a weight factor. We thus get
\begin{equation}
    \tilde\mu_k = \rho\tilde\mu_{k-1} + \sum_{n=2}^{k}  \tilde\mu_{n-1} \tilde\mu_{k-n}
\end{equation}
as claimed.
\end{proof}

We recall that $\mu_k$ are the moments of the
rescaled Marchenko-Pastur law:
\beq\label{eq:fdefn}
f(t)=\rho^{-2} f^{MP}_{\rho^{-1}}\l(\frac{t}{\rho}\r),
\eeq 
where $f^{MP}_{\rho^{-1}}$ is the probability density function of the Marchenko-Pastur law with parameter $\rho^{-1}$.

\be{lm}\label{lm:solve}
We have $\tilde \mu_k=\mu_k$ for all $k\geq 0$.
\e{lm}

\begin{proof}[Proof of Lemma \ref{lm:solve}]
First, we note that $\tilde\mu_0=\mu_0=1$ The moments $\mu_k^{MP}$ of the Marchenko-Pastur law $f^{MP}_{\rho}$ satisfy the recursion relation
\beq\label{eq:mprec}
\mu_k^{MP}=\rho \mu_{k-1}^{MP}+\sum_{n=2}^k \mu_{n-1}^{MP}\mu_{k-n}^{MP}.
\eeq
By a change of variables, we see that
$$
\mu_k=\int_\R t^k f(t)\d t = \rho^{-2} \int_\R t^k f^{MP}_{\rho^{-1}}\l(\frac{t}{\rho}\r)\d t
=\rho^{k-1} \int_\R t^k f^{MP}_{\rho^{-1}}(t)\d t=\rho\mu_k^{MP}.
$$
Combining this fact with \eqref{eq:mprec}, we see that the $\mu_k$ solve the recursion relation \eqref{eq:recursion} in Theorem \ref{thm:recursion}, with the same initial condition.
\end{proof}

\section{Verifying the quasi-random condition on frequencies}\label{sect:ntheory}
In this section, we use techniques from analytic number theory (estimates on exponential sums) to verify the quasi-randomness assumption for several classes of examples. These examples all have in common that sufficient irrationality prevents the frequencies from colluding to produce large exponential sums.

Throughout this section, we fix a choice of $\rho>0$. We recall the relevant exponential sum.
\beq\label{eq:ESNdefn}
ES_{N}(\ul{\omega}):=\frac{1}{N^5} \sum_{i_1,i_2=1}^{\floor{\rho N}} \sum^N_{\substack{j_1,j_2,j_3,j_4=1\\ j_1+j_3=j_2+j_4}} 
e\l[\frac{\om_{i_1}-\om_{i_2}}{2}(j_1^2-j_2^2+j_3^2-j_4^2)\r]
\eeq

We recall Definition \ref{defn:qrandom} of a quasi-random sequence. Let $\de,\rho>0$. The sequence of frequencies $\ul{\omega}=(\omega_1,\omega_2,\ldots)$ with $\omega_i\in [0,1]$ is $(\de,\rho)$-quasi-random, if there exists a constant $C>0$ such that we have the exponential sum bound
$$
|ES_N(\ul{\omega})|\leq C N^{-\de},\qquad \forall N\geq 1.
$$

\subsection{Preliminary analysis of $ES_N(\ul{\omega})$}
It will be convenient to reparametrize $ES_N(\ul{\omega})$ because this allows us to complete a square to reduce to a geometric series. This observation is closely related to a method known as Weyl differencing in analytic number theory \cite{Mont,W1,W2}; one point here is that the square is already present in the original exponential sum $ES_N$.  

\begin{lm}\label{lm:complete}
We have the identity
\beq\label{eq:complete}
ES_N(\ul{\omega})=\frac{1}{N^5} \sum_{i_1,i_2=1}^{\floor{\rho N}} \sum_{t=1-N}^{N-1}\l|\sum_{a=\max\{1,1-t\}}^{\min\{N,N-t\}} e\l[(\om_{i_1}-\om_{i_2})at\r]\r|^2.
\eeq
\end{lm}

We remark that this identity shows in particular that $ES_N(\ul{\omega})\geq 0$.

\begin{proof}
We change summation variables from $(j_1,j_2,j_3,j_4)$ subject to $j_1+j_3=j_2+j_4$ to the three independent variables $a,b,t$ defined by the relations
$$
j_1=a,\quad j_2=a+t,\quad j_3=b+t,\quad j_4=b.
$$
These variables allow us to complete a square:
$$
\begin{aligned}
ES_N(\ul{\omega})
=&\frac{1}{N^5} \sum_{i_1,i_2=1}^{\floor{\rho N}} \sum_{a=1}^N \sum_{t=\max\{1-a,1-b\}}^{\min\{N-a,N-b\}}
e\l[(\om_{i_1}-\om_{i_2})(b-a)t\r]\\
=&\frac{1}{N^5} \sum_{i_1,i_2=1}^{\floor{\rho N}} \sum_{t=1-N}^{N-1}\l|\sum_{a=\max\{1,1-t\}}^{\min\{N,N-t\}} e\l[(\om_{i_1}-\om_{i_2})at\r]\r|^2
\end{aligned}
$$
and this proves the lemma.
\end{proof}

Given Lemma \ref{lm:complete}, it is natural to perform the geometric series in $a$.

\be{cor}\label{cor:geometric} We have
\beq\label{eq:geometric}
ES_N(\ul{\omega})\leq \frac{1}{4N^5} \sum_{i_1,i_2=1}^{\floor{\rho N}} \sum_{t=1-N}^{N-1}\min\l\{N^2,\frac{1}{\|(\om_{i_1}-\om_{i_2})t\|^2_{\mathbb T}}\r\}.
\eeq
\e{cor}

\be{proof}
Consider the geometric sum in $a$ in \eqref{eq:complete}. It can be trivially bounded by $N$, or we can perform the sum. This yields
$$
\begin{aligned}
ES_N(\ul{\omega})=&\frac{1}{N^5} \sum_{i_1,i_2=1}^{\floor{\rho N}} \sum_{t=1-N}^{N-1}\l|\sum_{a=\max\{1,1-t\}}^{\min\{N,N-t\}} e\l[(\om_{i_1}-\om_{i_2})at\r]\r|^2\\
\leq& \frac{1}{4N^5} \sum_{i_1,i_2=1}^{\floor{\rho N}} \sum_{t=1-N}^{N-1}\min\l\{N^2,\frac{1}{\|(\om_{i_1}-\om_{i_2})t\|^2_{\mathbb T}}\r\}.
\end{aligned}
$$
Here we used that $\frac{1}{2}|\sin(\pi x)|\geq \|x\|_{\mathbb T}$, the distance of $x$ to the nearest integer. This proves Corollary \ref{cor:geometric}.
\end{proof}

We now verify the bound $|ES_N(\ul{\omega})|\leq CN^{-\de}$ in explicit examples.

\subsection{Frequencies generated by irrational circle rotation.}
In this section, we study the case $\om_i=i \al$ with $\al$ Diophantine, i.e., the frequencies are generated by an irrational circle rotation. We will control small divisors to prove the following proposition (which verifies the $(\de,\rho)$-quasirandom definition for this choice of $\om_i$).

\begin{prop}\label{prop:ialpha}
Let $\om_i=i \al$ with $\al\in (0,1)$ satisfying the Diophantine assumption 
\beq\label{eq:diophantine}
\|n\al\|_{\mathbb T}\geq\frac{C_\al}{n^p},\qquad \forall n\geq 1,
\eeq
for a power $p\geq 1$. Then, there exists a constant $C_\al'>0$ such that
$$
|ES_N(\ul{\omega})|\leq C'_{\al}N^{-1/(2p)}.
$$
\end{prop}

\be{rmk}
The set of admissible $\al$ in this proposition has full Lebesgue measure.
\e{rmk}

\be{proof}[Proof of Proposition \ref{prop:ialpha}]
For simplicity, we set $\rho=1$. The case of general $\rho$ follows by the same argument. We start from the bound in Corollary \ref{cor:geometric} and change variables to $r=i_1-i_2$ and $s=i_1+i_2$, and we note that for every choice of $r$, there exist at most $2N$ viable choices of $s$. This gives
$$
\begin{aligned}
ES_N(\ul{\omega})\leq& \frac{1}{4N^5} \sum_{i_1,i_2=1}^{\floor{\rho N}} \sum_{t=1-N}^{N-1}\min\l\{N^2,\frac{1}{\|(i_1-i_2)t\al\|^2_{\mathbb T}}\r\}\\
\leq & \frac{1}{2N^4} \sum_{r=1-N}^{N-1}  \sum_{t=1-N}^{N-1}\min\l\{N^2,\frac{1}{\|rt\al\|^2_{\mathbb T}}\r\}\\
\leq & \frac{2}{N^4} \sum_{r=0}^{N-1}  \sum_{t=0}^{N-1}\min\l\{N^2,\frac{1}{\|rt\al\|^2_{\mathbb T}}\r\}.
\end{aligned}
$$
In the last step, we reflection symmetry of the summand in $r$ and $t$ (and doubled up the contributions from $r=0$ and $t=0$).

Notice that due to the product structure, the combined sum over $r$ and $t$ and can be expressed in terms of the divisor function, defined by 
$$
\tau(n):=\sum_{d\vert n}1
$$
for every integer $n$. We will use the following bound on $\tau$ \cite{Apostol}.
$$
\forall \eps>0:\qquad \tau(n)=o(n^{\eps}).
$$
Changing variables to $n=rt$ then gives
\beq\label{eq:nalpha1}
\begin{aligned}
ES_N(\ul{\omega})
\leq & \frac{2}{N^4} \sum_{n=0}^{(N-1)^2}  \tau(n)\min\l\{N^2,\frac{1}{\|n\al\|^2_{\mathbb T}}\r\}\\
\leq&\frac{2C_\eps N^{2\eps}}{N^4} \sum_{n=0}^{(N-1)^2}  \min\l\{N^2,\frac{1}{\|n\al\|^2_{\mathbb T}}\r\}.
\end{aligned}
\eeq
where the constant $C_\eps$ comes from the divisor bound.

Next, we use the arithmetic structure of the set where $\frac{1}{\|n\al\|^2_{\mathbb T}}$ is large. We perform a dyadic decomposition in $n$. Given an integer $l\geq 0$ and a real number $x$, we write $x\sim 2^l$ for $2^l\leq x<2^{l+1}$. Define the set
$$
A_l:=\setof{0\leq n\leq (N-1)^2}{\frac{1}{\|n\al\|_{\mathbb T}}\sim 2^l}
$$
We write $|A_l|$ for the cardinality of $A_l$. 

We note that there exists a constant $C>0$ depending on $\al$ such that  $|A_l|\leq C\frac{N^2}{2^{l/p}}$. Indeed, by the Diophantine assumption \eqref{eq:diophantine}, we have
$$
n\in A_l\,\,\Rightarrow\,\, n\geq \l(\frac{C_\al}{2}\r)^{1/p} 2^{l/p}.
$$
and, by the triangle inequality,
$$
n,n'\in A_l\textnormal{ and } n<n'\,\,\Rightarrow\,\, n'-n\geq \l(\frac{C_\al}{4}\r)^{1/p} 2^{l/p}.
$$

We return to \eqref{eq:nalpha1} and define $L:=\log_2\left\lceil\l(\frac{(N-1)^2}{C_\al}\r)^p\right\rceil$. By $|A_l|\leq C\frac{N^2}{2^{l/p}}$, we have
$$
\begin{aligned}
ES_N(\ul{\omega})
\leq&\frac{8C_\eps N^{2\eps}}{N^4} \sum_{l=0}^{L} \l(\min\{N,2^{l}\}\r)^2 |A_l|\\
\leq&\frac{C_{\eps,\al}N^{2\eps}}{N^2} \sum_{l=0}^{L} \l(\min\l\{N,2^{l}\r\}\r)^2 2^{-l/p}\\
\leq &\frac{C_{\eps,\al}N^{2\eps}}{N^2} \l(\sum_{l=0}^{\log_2 N}  2^{2l-l/p}+\sum_{l=\log_2 N}^{\infty}  N^2 2^{-l/p}\r)\\
=&C_{\eps,\al} N^{2\eps-1/p}.
\end{aligned}
$$
Finally, we choose $\eps=1/(4p)$ in the divisor bound. This proves Proposition \ref{prop:ialpha}.
\e{proof}

\subsection{Fractional power frequencies}
We study frequencies $\omega_i$ which are obtained by taking fractional powers of $i$ modulo $1$. The tool we use here are the $k$-th derivative van der Corput estimates \cite{Mont,Titchmarsh} for exponential sums. A key difference to the previous example is that we now extract oscillations from the $i_1,i_2$ summations in \eqref{eq:ESNdefn}.

We also include the case where the powers are rescaled so that the frequencies are small, since it follows by the same proof method. In this case, the intuition is that small numbers $\ll N$ effectively look irrational on the scale of the matrix.

\begin{prop}[Power law frequencies]
\label{prop:vdc}
Let $\al,\beta\in \R\setminus\Z$ with $\al>\beta-2$. Define 
\beq
\om_i=\frac{i^\al}{N^\beta}.
\eeq
Then, we have $|ES_N(\ul{\om})|\leq CN^{-\de}$ for some $\de$ depending on $\al,\beta$.
\end{prop}

For instance, the exponential sum is subleading for the choices $\om_i=\sqrt{i}$ and $\om_i=\frac{1}{\sqrt{i}}$; see Figure \ref{fig:3/2}.


\begin{figure}[t]
\begin{center}
    \includegraphics[scale=.2]{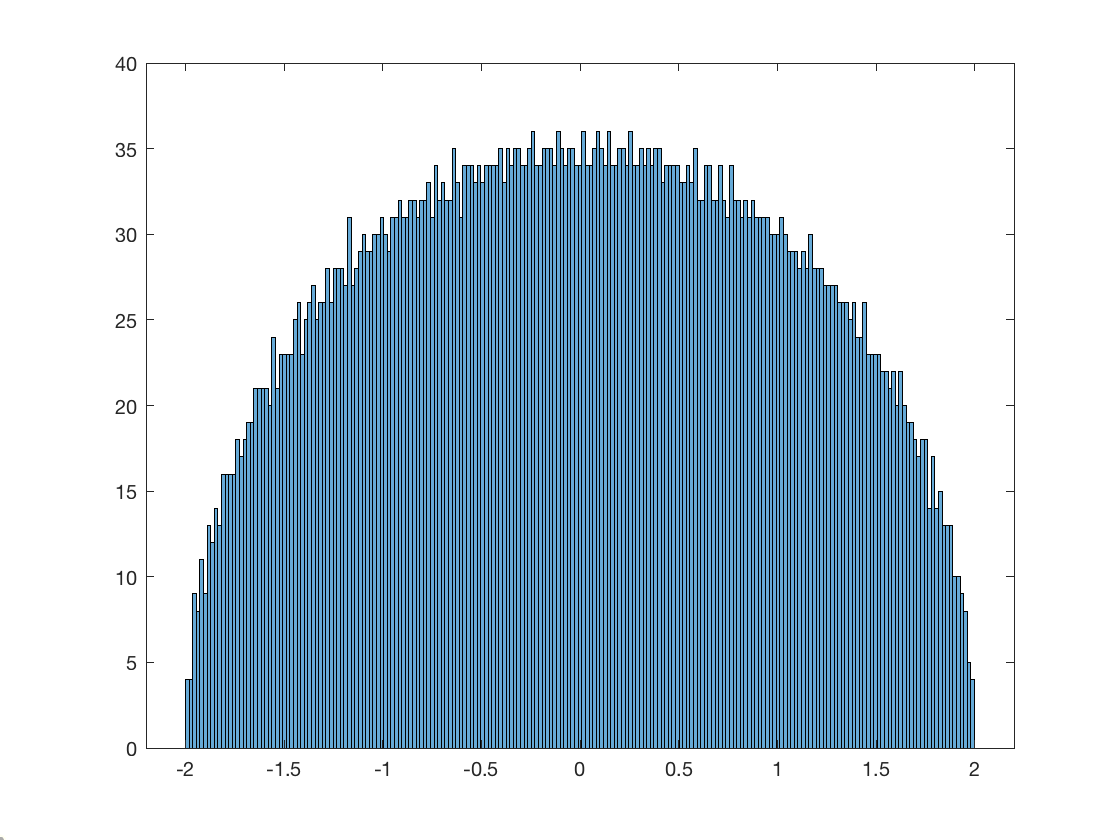}
    \caption{The empirical eigenvalue distribution $\frac{1}{N}\sum_{j=1}^N \delta_{\lam_j}$ of a $6000\times 6000$ matrix $H$ defined via \eqref{eq:Xdefn} with $\om_i=\sqrt{i}$ and $x_i=0$ for $i\geq 1$.}
    \label{fig:3/2}
\end{center}
\end{figure}

\begin{proof}[Proof of Proposition \ref{prop:vdc}]
For simplicity, we give the proof with $\rho=1$. The case of general $\rho$ follows by the same argument.
\textit{Step 1:}
We rewrite the exponential sum in Lemma \ref{lm:complete}, placing the $i_1$ and $i_2$ summations inside, because we now extract oscillations from these. Changing variables to $a-a'=r$ and $a+a'=s$, and performing the sum over $s$, we obtain
$$
\begin{aligned}
ES_N(\ul{\omega})
=&\frac{1}{N^5} \sum_{i_1,i_2=1}^{N} \sum_{t=1-N}^{N-1}\l|\sum_{a=\max\{1,1-t\}}^{\min\{N,N-t\}} e\l[(\om_{i_1}-\om_{i_2})at\r]\r|^2\\
\leq& \frac{C}{N^4} 
\sum_{t,r=-N}^N \l|\sum_{i=1}^{N}e\l[\om_{i}rt\r]\r|^2
\leq \frac{C}{N^4} 
\sum_{t,r=0}^N \l|\sum_{i=1}^{N}e\l[\om_{i}rt\r]\r|^2.\\
\end{aligned}
$$
Let $\eps=\de/2$ with $\de>0$ determined in step 2 below. We decompose the last sum as follows:
$$
\begin{aligned}
\frac{16}{N^4} 
\sum_{t,r=0}^N \l|\sum_{i=1}^{N}e\l[\om_{i}rt\r]\r|^2
\leq& \frac{CN^{2(1-\eps)}}{N^2} +\frac{C}{N^4}
\sum_{t,r=\floor{N^{1-\eps}}}^N \l|\sum_{i=1}^{N}e\l[\om_{i}rt\r]\r|^2\\
\leq& CN^{-\de}+\frac{C}{N^4}
\sum_{t,r=\floor{N^{1-\eps}}}^N \l|\sum_{i=\floor{N^{1-\eps}}}^{N}e\l[\frac{i^\al rt}{N^\beta}\r]\r|^2
\end{aligned}
$$
In the first step, we used the trivial bound. It remains to analyze the second term.\\

\textit{Step 2:} We first consider the case where $\al-\beta\in (-2,-1)$, since it involves the first-derivative van der Corput estimate, which is slightly different from the higher-order ones. The relevant phase function is $f(x):=\mathrm{sgn}(\al)\frac{x^\al rt}{N^\beta}$ with $N^{1-\eps}\leq r,t\leq  N$, and so $f'(x)=|\al|\frac{x^{\al-1} rt}{N^\beta}>0$ satisfies the estimate
$$
\frac{|\al|}{2} N^{(1-\eps)(\al+1)-\beta} <f'(x)<2|\al| N^{(\al+1)-\beta},
$$
(Note that for the lower bound we used that $\al\neq 1$, since $\al$ is not an integer.)
Notice that this implies in particular $f'(x)\ll 1$ by $\al+1<\beta$. Hence, we have the van der Corput estimate \cite{Mont,Titchmarsh}
$$
\begin{aligned}
ES_N(\ul{\omega})
\leq& CN^{-\de}+\frac{16}{N^4}
\sum_{t,r=\floor{N^{1-\eps}}}^N \l|\sum_{i=\floor{N^{1-\eps}}}^{N}e\l[\frac{i^\al rt}{N^\beta}\r]\r|^2\\
\leq& CN^{-\de}+CN^{-2-2(1-\eps)(\al+1)+2\beta}. 
\end{aligned}
$$
Recall that $\eps=\de/2$. We define $\de$ as the solution to $\de=2+2(1-\de/2)(\al+1)-2\beta$, that is,
$$
\de=2\l(1-\frac{\beta}{2+\al}\r)
$$
which is strictly positive for $\al-\beta\in (-2,-1)$. This proves $ES_N(\ul{\omega})\leq CN^{-\de}$ for $\al-\beta\in (-2,-1)$.\\

\textit{Step 3:} We consider the case where $\al-\beta\geq -1$. We set $k:=\floor{\al-\beta+3}$ and note that $k\geq 2$. The phase function $f(x)=\mathrm{sgn}(\al)\frac{x^\al rt}{N^\beta}$ derived in step 1 satisfies the derivative bound
$$
C_k N^{(1-\eps)(\al+2-k)-\beta} <f^{(k)}(x)<C_k' N^{\al+2-k-\beta}
$$
for constants $C_k,C_k'>0$ depending on $\rho,k,\alpha,\beta$, but not on $N$. We can use this to apply the $k$-derivative van der Corput estimate \cite{Mont,Titchmarsh} to our above bound on $ES_N(\ul{\omega})$. This gives
$$
\begin{aligned}
&\frac{16}{N^4}
\sum_{t,r=\floor{N^{1-\eps}}}^N \l|\sum_{i=\floor{N^{1-\eps}}}^{N}e\l[\frac{i^\al rt}{N^\beta}\r]\r|^2\\
&\leq C\frac{1}{N^2}\l(N^{1+2^{2-k}\eps(\al+2-k)} N^{\frac{(1-\eps)(\al+2-k)-\beta}{2^k-2}}
+CN^{1-2^{2-k}} N^{-\frac{(1-\eps)(\al+2-k)-\beta}{2^k-2}}\r)^2\\
&\leq C N^{2\l(2^{2-k}\eps(\al+2-k)+\frac{(1-\eps)(\al+2-k)-\beta}{2^k-2}\r)}
+CN^{2\l(-2^{2-k}-\frac{(1-\eps)(\al+2-k)-\beta}{2^k-2}\r)}.
\end{aligned}
$$
We recall that $\eps=\de/2$ and $k=\floor{\al-\beta+3}$, and so
$$
\al+2-k-\beta=\al-\beta+3-\floor{\al-\beta+3}-1\in [-1,0).
$$
This fact implies that at $\eps=0$ the exponents above are strictly negative, i.e.,
$$
\min\l\{\frac{(\al+2-k)-\beta}{2^k-2},-2^{2-k}-\frac{(\al+2-k)-\beta}{2^k-2}\r\}<0.
$$
(For the second exponent, this uses $\frac{1}{2^k-2}\leq 2^{2-k}$ which is equivalent to $k\geq 2$, and this is the reason why we treated the $k=1$ case separately in step 2 above.) Hence, the intermediate value theorem implies that we can choose $\de>0$ such that $ES_N(\ul{\omega})\leq CN^{-\de}$. This proves Proposition \ref{prop:vdc}.
\end{proof}

\subsection{Generic frequencies.}
In this section, we prove the following proposition.

\begin{prop}[Generic frequencies]\label{prop:generic}
Fix two integers $N$ and $M=\floor{\rho N}$. Then, a frequency vector $(\om_1,\ldots,\om_M)\in [0,2]^M$ is $(\de,\rho)$ quasi-random with high probability with respect to Lebesgue measure, that is,
$$
\mathbb P\l(\setof{(\om_1,\ldots,\om_M)\in [0,2]^M}
{(\om_1,\ldots,\om_M)\textnormal{ is not $(\de,\rho)$-quasi-random}}\r)< \frac{1}{N^{1-\de}}.
$$
\end{prop}
 
 This estimate can be improved by studying higher moments (here we just bound the first moment).
 
\be{rmk}
The choice of sampling $(\om_1,\ldots,\om_M)$ over  $[0,2]^M$ instead of $[0,1]^M$ is an artifact of our derivation of $ES_N(\ul{\om})$ by averaging the skew-shift. We emphasize that it is inconsequential for the purpose of constructing the matrices $H_{M,N}$ in our main result, since each matrix entry $X_{i,j}=e\l[\binom{j}{2}\om_i+jy_i+x_i\r]$ does not change under shifting $\om_i$ by $1$.
\e{rmk}

\begin{proof}
We integrate $(\om_1,\ldots,\om_M)$ over $[0,2]^M$ in Definition \eqref{eq:ESNdefn}. By orthonormality of the family $\{e[j\cdot]\}_{j\in\Z}$, we obtain
$$
\int_0^2\ldots \int_0^2 ES_N(\om_1,\ldots,\om_M)\d \om_1 \ldots\d\om_M 
\leq 
\frac{4M^2}{N^5} \sum_{\substack{1\leq j_1,j_2,j_3,j_4\leq N\\ j_1+j_3=j_2+j_4\\ j_1^2+j_3^2=j_2^2+j_4^2}} 1. 
$$
We thus have to count the number of solutions to the simple system of Diophantine equations
$$
j_1+j_3=j_2+j_4 \,\textnormal{ and }\, j_1^2+j_3^2=j_2^2+j_4^2.
$$
By subtracting the first equation from the second, we find $j_1-j_3=\pm(j_2-j_4)$. Supposing without loss of generality that $j_2\geq j_4$, and adding this identity to the first equation, we find $j_1=j_2$ and $j_3=j_4$. We conclude that
$$
\int_0^2\ldots \int_0^2 ES_N(\om_1,\ldots,\om_M)\d \om_1 \ldots\d\om_M 
\leq 
\frac{16M^2}{N^3} \leq \frac{16\rho^2}{N}. 
$$
By Markov's inequality, this proves Proposition \ref{prop:generic}.
\end{proof}

\section{Outlook: Deterministic matrices}\label{sect:deterministic}
A natural follow-up question to our main result is to ask for a completely deterministic matrix whose global eigenvalue distribution is semicircular. Our model presented in the main results contains $N$ random variables $y_1,\ldots,y_N$ chosen uniformly and independently from $[0,1]$. Their presence is instrumental for our proof, since it ensures the Kirchhoff circuit law. This reduces the main technical step to verifying some cancellation in the relevant exponential sums, as opposed to having to study their precise asymptotics in $N$. By contrast, for completely deterministic matrices, the Kirchhoff current law is no longer available and consequently the situation is much more delicate. 

In this section, we present some preliminary findings regarding the eigenvalues of certain fully deterministic matrices generated whose entries are generated by the toroidal shift or skew-shift. The main take-away from these examples is that the semicircle law can no longer be expected in general for deterministic matrices, and if it arises, it is accompanied by heavy tails which render the moment method ineffective.

\subsection{Overview of deterministic models}
In this section, we present our findings towards completely deterministic matrices. We define $3$ deterministic models, called A, B, and C, which are natural variations of the skew-shift models considered in this paper. We first refer the reader to Table \ref{table}, where models A, B, and C are defined and our numerical findings are summarized. Note that model C is a shift, not a skew-shift model.

The third column of the table shows the moments. In the deterministic setting, these are defined as
$$
\mu^{(2k)}_{N}=\frac{1}{2N}\Tr H^{(2k)}
$$
where $H_{N,N}$ is defined by \eqref{eq:Hdefn}. Our comments on the findings in Table \ref{table} are as follows.

\newcolumntype{Y}{>{\centering\arraybackslash}X}
\newcolumntype{s}{>{\hsize=.3\hsize}Y}
\captionsetup{width=12cm}

\begin{table}[t]
\begin{center}
\begin{tabularx}{\textwidth}{ c | Y | s }
 Model  & Empirical spectral measure (normalized) & Moments \\
 \hline
 \textbf{A: } $j^2 \sqrt{i}$ 
        &  \begin{minipage}{.8\textwidth}
      \includegraphics[scale=0.19]{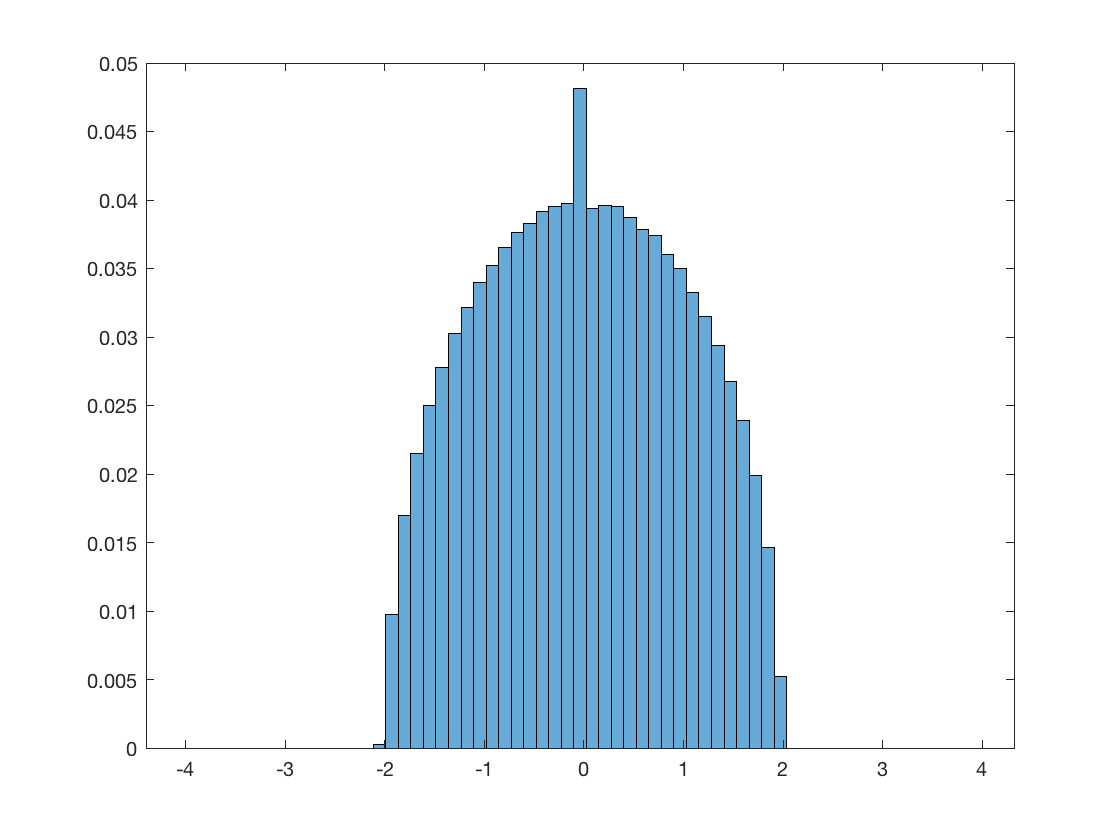}
    \end{minipage}
     &
    \vspace{-1cm}
    $\mu_{N}^{(2)}= 1\newline \mu_{N}^{(4)}\approx 3\newline \mu_{N}^{(6)}\approx 70
    \newline \mu_{N}^{(8)}\approx 4,000$
       \\
       \hline
   \textbf{B: } $j^2 i\sqrt{2}$
        &  \begin{minipage}{.8\textwidth}
      \includegraphics[scale=0.19]{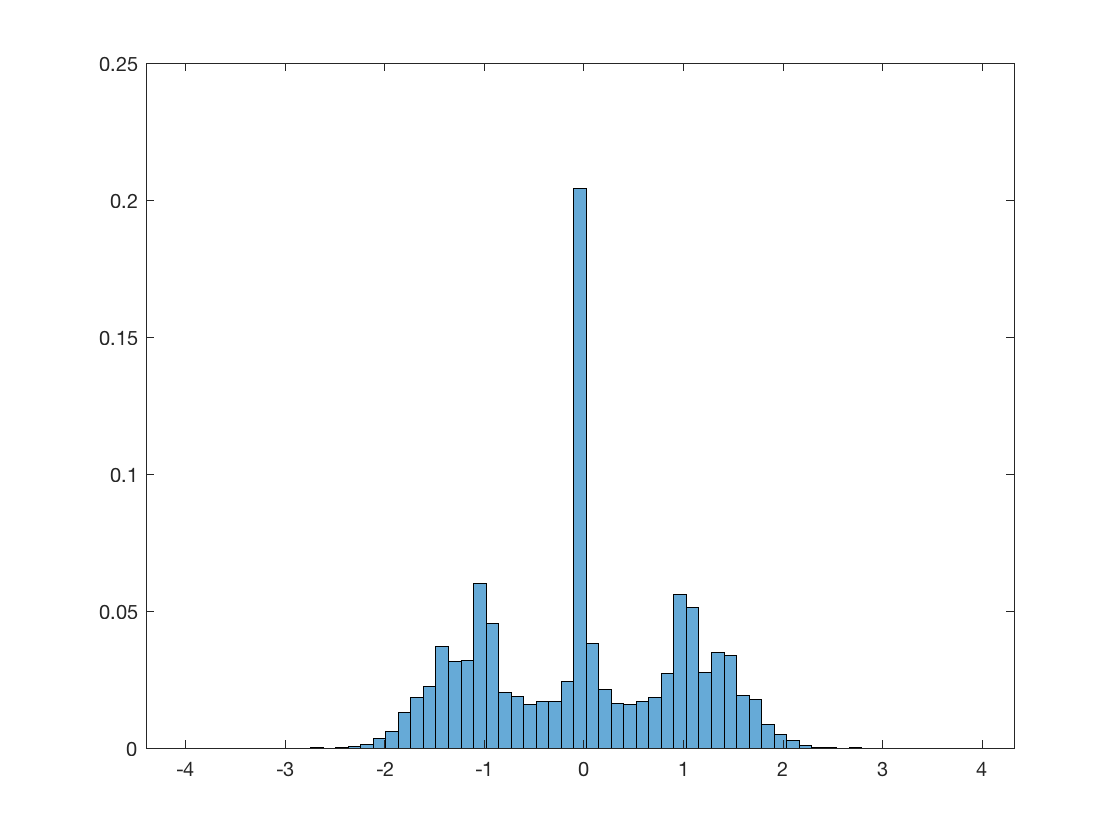}
    \end{minipage}
    &
    \vspace{-1cm}
    $\mu_{N}^{(2)}= 1\newline \mu_{N}^{(4)}\approx 2\newline \mu_{N}^{(6)}\approx 5
    \newline \mu_{N}^{(6)}\approx 16$
    \\
    \hline
 \textbf{C: } $j\sqrt{i}$ 
      &  \begin{minipage}{.8\textwidth}
      \includegraphics[scale=0.19]{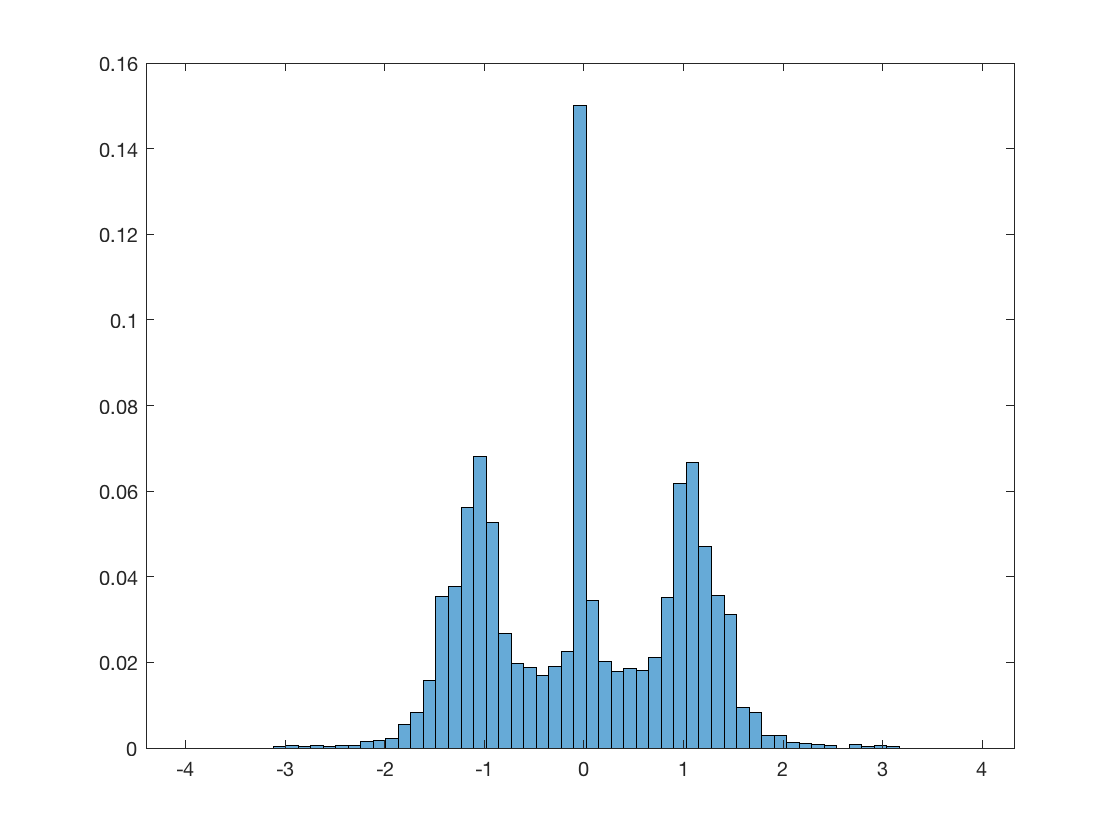}
    \end{minipage}
    &
    \vspace{-1cm}
    $\mu_{N}^{(2)}= 1\newline \mu_{N}^{(4)}\approx 3\newline \;\mu_{N}^{(6)}\approx 70
    \newline \mu_{N}^{(8)}\approx 4,000$
\end{tabularx}
\caption{\label{table}
The models A, B, and C are deterministic matrices $H_{N,N}$ whose entries are defined as $X_{i,j}=e[\cdot]$ with $[\cdot]$ given by the expressions above. The plots show the empirical spectral measure of $8000 \times 8000$ matrices $H_{N,N}$.}
\end{center}
\end{table}

\be{itemize}
\item Remarkably, model A still displays a semicircle law, but with heavy tails which make the moments different from the Catalan numbers and render the moment method ineffective. One may conjecture that a local semicircle law holds for model A, but the required estimates for exponential sums would be delicate. (In particular, the exact order of fluctuations would need be analyzed precisely.) 
\item For models B and C, we observe that the empirical spectral measure follows a novel bimodal distribution and thus differs significantly from a semicircle law. The bimodal distributions of models B and C are similar, but distinct. It is unclear at this stage if there is a universal bimodal distribution that arises as the limiting distribution of a variety of models. Understanding the limiting distribution more precisely is an interesting open problem and presumably involves good understanding of small denominators.
\item Models A and C both display rather large extreme eigenvalues ($\approx 7$ at the considered matrix size of $8000 \times 8000$). These heavy tails are matched by the significant size of their moments. Surprisingly, the first few moments of models A and C do not differ by very much.
\item For model B, the numerical moments appear to be very close to the semicircle law, but below we prove that this is spurious (see Theorem \ref{thm:modelb}).
\e{itemize}

Here we focused on the case of square matrices for simplicity. The models and the results can be generalized to the rectangular case $M=\floor{\rho N}$; see also the remark after Theorem \ref{thm:modelb}. The number $\sqrt{2}$ in the definition of model B can be replaced by any irrational number satisfying a Diophantine condition without changing the qualitative results.\\

We also consider the empirical eigenvalue spacing distribution numerically. We observe numerically that model A exhibits the level spacing distribution of GUE matrices. 

Fix an energy $E\in(-2,2)$ and a cutoff parameter $t<1$ with $Nt\to \infty$. For a model of the form \eqref{eq:Hdefn} whose spectral distribution follows the Wigner semicircle law, we then have the following definition of the empirical cumulative distribution function of the level spacing near $E$:
$$
\Lam_N(s):=\frac{1}{4Nt\rho_{sc}(E)}
\l|\setof{1\leq j\leq N-1}{\lam_{j+1}-\lam_j\leq \frac{s}{2N\rho_{sc}(E)},\; |\lam_j-E|\leq t}\r|.
$$
The level spacing distribution of a GUE matrix is approximately given by the Wigner surmise function $W(s)=32 \pi^{-2}s^2 e^{-4s^2/\pi}$. The level spacing of GUE matrices is known to be universal among a large class of Hermitian random matrices. Numerically, it appears that model A belongs to this universality class as shown in Figure \ref{fig:spacing}. This remains true for some natural variants of model A, for instance, if one replaces $\sqrt{i}$ in the definition of the matrix by other factional powers of $i$, for example $i^{1/3}$.\\

\begin{figure}[t]
    \centering
    \includegraphics[scale=0.2]{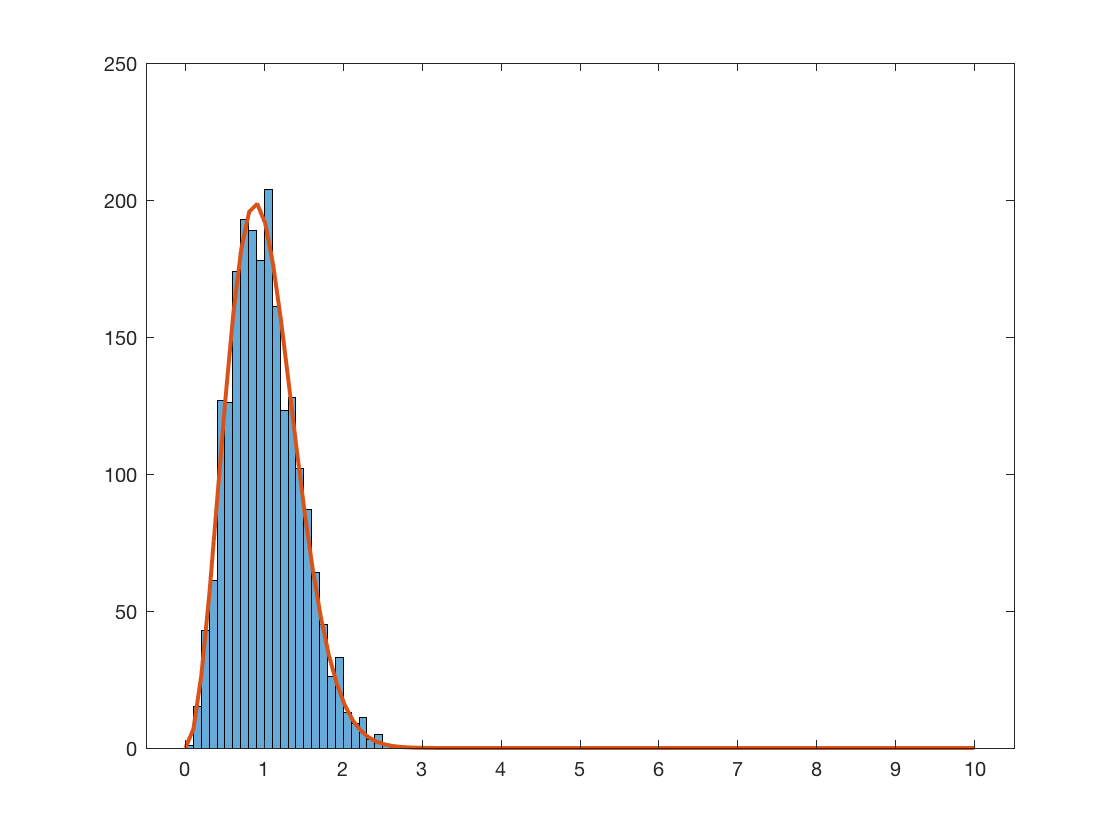}
    \caption{The histogram gives the empirical eigenvalue spacing distribution for model A with $N=4,000$ and cutoff parameter $t=N^{-1/10}$. The red curve is the Wigner surmise function $W(s)=32 \pi^{-2}s^2 e^{-4s^2/\pi}$.}
    \label{fig:spacing}
\end{figure}

We can summarize our numerical findings presented in this section as follows: The dynamics underlying model A (and its variants) are sufficiently quasi-random that the resulting Hermitian matrices still display some of the spectral features of GUE matrices. This is remarkable insofar as these matrices are fully deterministic. On the other hand, models B and C do not appear to be sufficiently quasi-random for this to be the case. We intuitively ascribe this difference to the presence of a linear term in models B and C, which corresponds to a more regular generating dynamics. Understanding these connections between dynamics and spectral theory rigorously is an open problem.

In the following section, we answer a question that is raised by the data in Table \ref{table} by analytical methods.

\subsection{Analysis of the deterministic models}

We can analyze models A-C with another graphical representation formula for the moments, which is a deterministic analog of Proposition \ref{prop:graphrep}. It holds for any deterministic matrix model $X_{i,j}$ with the deterministic propagator
$$
\curly{K}_{(i,i')}(j):=N X_{i,j}X_{j,i'}^*.
$$

\be{prop}[Deterministic graphical representation formula]\label{prop:graphrepdet}
We have
\beq
\mu^{(2k)}_{N,N}
=\frac{1}{N^{1+k}}\sum_{G_L=(V,L)\in \curly{L}_k}  \sum_{\substack{\ul{i}=(i_1,\ldots,i_k)\\ L_{\ul{i}}\sim L}}\sum_{\ul{j}=(j_1,\ldots,j_k)} \prod_{r=1}^k \curly{K}_{(i_r,i_{r+1})}(j_r).
\eeq
\e{prop}

Notice that the $L$-admissible condition for the currents $\ul{j}$ is dropped now, this means that the Kirchhoff circuit law is no longer enforced.

\begin{proof}
The claim follows in the same way as formula \eqref{eq:graphical0} in Proposition \ref{prop:graphrepdet}. There is no averaging step to consider afterwards.
\end{proof}

Notice that the first few moments for model B are close to the first few Catalan numbers. We can show that the fourth moment is, however, strictly distinct from $c_2=2$.

\begin{thm}
\label{thm:modelb}
For model B, there exists a constant $\eps>0$ such that
$$
\mu^{(4)}_{N,N}\geq 2+\eps,
$$
along a subsequence of $N\to \infty$.
\end{thm}

\be{rmk}
The same proof works if $\sqrt{2}$ is replaced by any Diophantine number $\al$. The proof also works for rectangular matrices where $M=\floor{\rho N}$ with any $\rho>0$, and in that case it shows that the second moment is strictly larger than the second of the Marchenko-Pastur distribution. Moreover, when $\rho<1/2$, the argument can be strengthened to apply for all $N$ large enough.
\e{rmk}

\begin{proof}
Let $X_{i,j}=N^{-1/2}e[\sqrt{2} ij^2]$, so that
$$
\curly{K}_{(i,i')}(j)=e[\sqrt{2} (i-i') j^2].
$$
We use Proposition \ref{prop:graphrepdet} with $k=2$. There are exactly two distinct explorations on $2$ edges:
$$
L_1=((1,2),(2,1)),\qquad L_2=((1,1),(1,1)).
$$
Since $\curly{K}_{i\to i}(j)=1$, the contribution from $L_2$ is $1$. It remains to consider the contribution from $L_1$, which is
$$
\begin{aligned}
\Phi(L_1):=&\frac{1}{N^3} \sum_{i_1,i_2=1}^{N} \sum_{j_1,j_2=1}^N e[\sqrt{2} (i_1-i_2) (j_1^2-j_2^2)]\\
=&1+\frac{1}{N^3}  \sum_{\substack{j_1,j_2=1\\ j_1\neq j_2}}^N \l|\sum_{i=1}^N e[\sqrt{2} i (j_1^2-j_2^2)]\r|^2\\
=&1+\frac{1}{N^3}  \sum_{\substack{j_1,j_2=1\\ j_1\neq j_2}}^N \l|\frac{\sin(\pi N \sqrt{2} (j_1^2-j_2^2))}{\sin(\pi \sqrt{2} (j_1^2-j_2^2))}\r|^2\\
\geq& 1+\frac{4}{\pi^2 N^3}  \sum_{\substack{j_1,j_2=1\\ j_1\neq j_2}}^N \l(\frac{\| N \sqrt{2} (j_1^2-j_2^2)\|_\T}{\| \sqrt{2} (j_1^2-j_2^2)\|_\T}\r)^2
\end{aligned}
$$
In the last step, we used that $\|x\|_\T\leq \frac{1}{2}|\sin(\pi x)|\leq \frac{\pi}{2}\|x\|_\T$ for all real numbers $x$. 

In view of the claim, it suffices to show that the last sum is bounded below along a subsequence of $N\to\infty$.

\begin{lm}\label{lm:N'}
There exists $\de'>0$ such that for $N'\in \{N,\floor{8N/7}\}$, it holds that
$$
\frac{4}{\pi^2 (N')^3} \sum_{\substack{j_1,j_2=1\\ j_1\neq j_2}}^{N'} \l(\frac{\| 4N' \sqrt{2} (j_1^2-j_2^2\|_\T}{\| 4\sqrt{2} (j_1^2-j_2^2)\|_\T}\r)^2\geq \de'>0
$$
\end{lm}

Notice that Lemma \ref{lm:N'} implies Theorem \ref{thm:modelb}. We now prove this lemma. We change variables to 
$$
s:=j_1+j_2,\qquad r:=j_1-j_2.
$$
Noting that $j_1^2-j_2^2=rs$, this gives
\beq\label{eq:evenodd}
\begin{aligned}
&\frac{4}{\pi^2 N^3} \sum_{\substack{j_1,j_2=1\\ j_1\neq j_2}}^N \l(\frac{\| M \sqrt{2} (j_1^2-j_2^2)\|_\T}{\| \sqrt{2} (j_1^2-j_2^2)\|_\T}\r)^2\\
&=\frac{8}{\pi^2 N^3} \sum_{s=2}^{2N} \sum_{\substack{1\leq r\leq \min\{s-2,2N-s\}\\ r-s\equiv 0\mod 2}} \l(\frac{\| N \sqrt{2} rs\|_\T}{\| \sqrt{2} rs\|_\T}\r)^2\\
&\geq \frac{8}{\pi^2 N^3} \sum_{s=2}^{2N} \sum_{r=1}^{\min\{s-2,2N-s\}} \ind_{\textnormal{$r$ and $s$ even}}\l(\frac{\| N \sqrt{2} rs\|_\T}{\| \sqrt{2} rs\|_\T}\r)^2\\
&\geq \frac{8}{\pi^2 N^3} \sum_{s=1}^{\floor{N/2}} \sum_{r=1}^{s} \l(\frac{\| 4N \sqrt{2} rs\|_\T}{\| 4\sqrt{2} rs\|_\T}\r)^2\\
&\geq \frac{4}{\pi^2 N^3} \sum_{s=1}^{\floor{N/2}} \sum_{r=1}^{\floor{N/2}} \l(\frac{\| 4N \sqrt{2} rs\|_\T}{\| 4\sqrt{2} rs\|_\T}\r)^2
\end{aligned}
\eeq
In the last step, we symmetrized in $r$ and $s$. Given $0<\eps<1$ we define the set of mass-containing pairs, a subset of the pairs of integers $(s,r)$ with $1\leq s,r\leq \floor{N/2}$ defined by
$$
\curly{M}:=\setof{(s,r)}{\|4\sqrt{2}rs\|_\T\in \l(0,\frac{3}{N}\r)}.
$$

By the pigeonhole principle, we know that there are a macroscopic number of pairs $(r,s)$ in $\curly{M}$.

\begin{lm}\label{lm:Mcardinality}
The cardinality of $\curly{M}$ is bounded below by $\floor{N/2}$.
\end{lm}

\begin{proof}[Proof of Lemma \ref{lm:Mcardinality}]
Fix an $1\leq s\leq \floor{N/2}$. Consider the collection of $\floor{N/2}$ numbers
$$
\setof{\|4\sqrt{2}rs\|_\T}{1\leq r\leq \floor{N/2}}\subset [0,1].
$$
By the pigeonhole principle, there exist $1\leq r_1< r_2\leq \floor{N/2}$ such that 
$$
|\|4\sqrt{2}r_2s\|-\|4\sqrt{2}r_1s\||< \frac{3}{N}.
$$
Then we can consider their difference $\tilde r=r_2-r_1$ and notice that $4\sqrt{2}\tilde{r} s=4\sqrt{2}r_2s-4\sqrt{2}r_1s=K+x$ where $K$ is an integer and $|x|<\frac{3}{N}$. Hence, $(s,\tilde r)\in \curly{M}$. Since $s$ was arbitrary, this shows $|\curly{M}|\geq \floor{N/2}$ as claimed. This proves the lemma.
\end{proof}

We return to the proof of Lemma \ref{lm:N'}. We distinguish two cases. Fix $0<\eps<\frac{1}{3}$. We define the bad set
$$
\curly{B}_\eps:=\setof{(s,r)\in\curly{M}}{\|4\sqrt{2}rs\|_\T\in \bigcup_{k=1}^3\l[\frac{k-\eps}{N},\frac{k+\eps}{N}\r]}.
$$

\textit{Case 1:} Assume that $|\curly{B}_\eps|< N/4$. Notice that for every $(r,s)\in\curly{M}\setminus \curly{B}_\eps$, we can write $4\sqrt{2}rs=K+u$ for some integer $K$ and some remainder 
$u\in (0,3/N)$ satisfying $|u-k/N|> \eps/N$ for $k=1,2,3$. Notice that then $4N\sqrt{2}rs=NK+Nu$ with $\|4N\sqrt{2}rs\|_\T=\|Nu\|_\T>\eps$. In conlusion, for $(r,s)\in\curly{M}\setminus \curly{B}_\eps$ we have
$$
\frac{\| 4N \sqrt{2} rs\|_\T}{\| 4\sqrt{2} rs\|_\T}\geq \frac{N\eps}{3}
$$
We then estimate the last expression in \eqref{eq:evenodd} by
$$
\begin{aligned}
\frac{4}{\pi^2 N^3} \sum_{s=1}^{\floor{N/2}} \sum_{r=1}^{\floor{N/2}} \l(\frac{\| 4N \sqrt{2} rs\|_\T}{\| 4\sqrt{2} rs\|_\T}\r)^2
&\geq \frac{4}{\pi^2 N^3} \sum_{(r,s)\in\curly{M}\setminus \curly{B}_\eps} \l(\frac{\| 4N \sqrt{2} rs\|_\T}{\| 4\sqrt{2} rs\|_\T}\r)^2\\
&\geq \frac{4}{\pi^2 N^3}\l(\frac{N\eps}{3}\r)^2 |\curly{M}\setminus \curly{B}_\eps|.
\end{aligned}
$$
Since we assumed that $|\curly{B}_\eps|< N/4$, Lemma \ref{lm:Mcardinality} implies that $|\curly{M}\setminus \curly{B}_\eps|\geq N/4$ and it follows that
$$
\begin{aligned}
\frac{4}{\pi^2 N^3} \sum_{s=1}^{\floor{N/2}} \sum_{r=1}^{\floor{N/2}} \l(\frac{\| 4N \sqrt{2} rs\|_\T}{\| 4\sqrt{2} rs\|_\T}\r)^2\geq \de>0.
\end{aligned}
$$
This proves the claim of Lemma \ref{lm:N'} in case 1, assuming $\de'$ is chosen $\leq \de$.

\textit{Case 2:} Assume that $|\curly{B}_\eps|\geq  N/4$. We define $\tilde N=\floor{8N/7}$. The idea is that pairs in the bad set $\curly{B}_\eps$, which are too close to a multiple of $1/N$, are in fact good points for $\tilde N$. Indeed, let $(r,s)\in \curly{B}_\eps$. Then we claim that, for $\eps>0$ sufficiently small,
\beq\label{eq:nowgood}
\|4\sqrt{2}rs\|_\T\in \l(0,\frac{4}{\tilde N}\r)\setminus \bigcup_{l=1}^4\l[\frac{l-\eps}{\tilde N},\frac{l+\eps}{\tilde N}\r].
\eeq
To see this, we write $\|4\sqrt{2}rs\|_\T=\frac{k}{N}+u$ with $k\in\{1,2,3\}$ and $|u|\leq \frac{\eps}{N}$. By the choice of $\tilde N=\floor{8N/7}$, we see that \eqref{eq:nowgood} can now be ensured for $\eps$ sufficiently small. 

By our assumption, the set of pairs $(r,s)$ such that \eqref{eq:nowgood} holds has cardinality $\geq N/4$. For every such pair, we can follow the argument given in case 1 to conclude that
$$
\begin{aligned}
\frac{4}{\pi^2 \tilde N^3} \sum_{s=1}^{\floor{\tilde N/2}} \sum_{r=1}^{\floor{\tilde N/2}} \l(\frac{\| 4\tilde N \sqrt{2} rs\|_\T}{\| 4\sqrt{2} rs\|_\T}\r)^2\geq \tilde \de>0
\end{aligned}
$$
Setting $\de':=\min\{\de,\tilde\de\}$, we see that Lemma \ref{lm:N'} is proved. 
\end{proof}

\begin{appendix}
\section{On Kirchhoff's current law}
In this appendix, we study the number of solutions to Kirchhoff's law. We will study the number of free parameters using vector space theory and so it is convenient to introduce $\R$-valued currents/edge weights. For this, we recall the notation from Definition \ref{defn:Feynman}

\begin{defn}[$\R$-valued edge weights]
Let $L$ be an exploration and $G_L=(V,L)$ its associated exploration graph.
\be{enumerate}[label=(\roman*)]
\item  Given a sequence $\ul{j}=(j_1,\ldots, j_k)\subset \R^k$, we assign the weight $j_i$ to the edge $(\nu_i,\nu_{i+1})$ in $L$.

We say that the sequence $\ul{j}$ is an admissible collection of edge weights for $L$ (or ``$L$-admissible'' for short), if the Kirchhoff circuit law holds on $G_L$, i.e., if
$$
\sum_{e \in I_v} j_e = \sum_{e \in O_v} j_e,\qquad \forall v\in V.
$$
\item Define the set
$$
\curly{C}:=\setof{\ul{j}=(j_1,\ldots, j_k)\subset \R^k}{\ul{j} \textnormal{ is $L$-admissible}}
$$
\e{enumerate}
\end{defn}

\be{rmk}
We may equivalently interpret a negative current $j_i<0$ as a positive current running in the opposite direction (which amounts to reorienting the corresponding edge).
\e{rmk}

We note that $\curly{C}$ is a vector space, since it is a subset of $\R^k$ defined through linear constraints containing the origin.

\subsection{Number of free parameters in Kirchhoff's law}
The key result of this section is

\be{thm}[Free parameters in Kirchhoff's law]
\label{thm:K}
Let $L$ be an exploration on $l$ vertices and $k$ edges. Then $\mathrm{dim}_\R (\curly{C})=k+1-l$.
\e{thm}


%

We will now prove Theorem \ref{thm:K}. Recall that an exploration graphs is endowed with an orientation of the edges.

\be{defn}\label{defn:curlyc}
Let $G=(V,E)$ be a connected directed graph. To each cycle in the underlying undirected graph of $G$, $$C = e_1 \rightarrow e_2 \rightarrow e_3\to\ldots\rightarrow e_j \rightarrow e_1,$$ we associate the vector $\delta_C= \sum_{i=1}^{j} \mathrm{sgn}(e_i) \delta_{e_i}$ where $\mathrm{sgn}(e_i)$ is 1 if the direction of $e_i$ along the cycle follows its orientation in the directed graph $G$ and -1 otherwise. We define the auxiliary vector space
$$
\tilde{\mathcal C}(G):=\text{span}\{\delta_C: C \text{ is a cycle in } G\}
$$
\e{defn}

\begin{lm}\label{lm:cyclespace}
Let $G=(V,E)$ be a directed graph on $k$ edges and $l$ vertices. Then $\mathrm{dim}_\R(\tilde{\mathcal C}(G))=k-l+1$. 
\end{lm}

\begin{lm}\label{lm:iso}
There is an isomorphism between the vector spaces $\tilde{\mathcal C}(G)$ and $\mathcal C$.
\end{lm}

\be{proof}[Proof of Theorem \ref{thm:K} assuming the lemmas]
Let $L$ be an exploration on $k$ edges and $l$ vertices, and associate to it the oriented exploration graph $G$. Thanks to Lemma \ref{lm:cyclespace}, it suffices to prove that
\beq\label{eq:dimeq}
\mathrm{dim}_\R(\tilde{\mathcal C}(G))=\mathrm{dim}_\R(\mathcal C).
\eeq
This follows from Lemma \ref{lm:iso}.
\end{proof}

It remains to prove the two lemmas.

\begin{proof}[Proof of Lemma \ref{lm:cyclespace}]
This will be a result of induction on $m= k-l +1$ . The base case is $m=0$. Since the graph is connected, this implies that the graph is a tree and the dimension of the vector space of all cycles is 0.

Now let us assume that the theorem holds for all values of $m \le n$. We will show it true when $m = n+1$. Since the graph $G$ is not a tree, there exists a cycle $C$ and we fix an edge $e\in C$. Consider the connected graph $G\setminus e$. By the induction hypothesis, we can find a basis $\{c_1,...,c_{n}\}$ for the space of cycles $\tilde{\mathcal{C}}(G\setminus e)$. We claim that $\{c_1,...,c_m, \delta_{C}\}$ is a basis for $\tilde{\mathcal{C}}(G)$. Indeed, if we have an element of $\tilde{\mathcal{C}}(G)$ that does not involve the edge $e$ then it will be an element of  $\tilde{\mathcal{C}}(G\setminus e)$ and, thus, will be in the span of $\{c_1,...,c_m\}$. Now consider an element of $\tilde{\mathcal{C}}(G)$ that does involve the edge $e$. Taking either the sum or difference of it with $\delta_{C}$ will produce an element of $\tilde{\mathcal{C}}(G)$ that avoids the edge $e$ and hence lies in $\tilde{\mathcal{C}}(G\setminus e)=\mathrm{span}\{c_1,...,c_m\}$. This proves that $\{c_1,...,c_m, \delta_{C}\}$ is a basis for $\tilde{\mathcal{C}}(G)$, and so this set has dimension $m+1$ as claimed. This proves Lemma \ref{lm:cyclespace}.
\end{proof}

\be{proof}[Proof of Lemma \ref{lm:iso}]
We fix an arbitrary element $\ul{j} \in \mathcal{C}$. From our graph $G$ remove all edges $(\nu_i,\nu_{i+1})$ such that $j_i=0$. Additionally, reorient all the edges of $G$ so that each edge has positive current running through it; call the resulting current $\tilde{\ul{j}}$ and the resulting graph $\tilde{G}$. Note that $\tilde{\ul{j}}$ satisfies Kirchhoff's current law for $\tilde{G}$. Find any directed cycle in $\tilde{G}$; such a cycle must exist or else Kirchhoff's current law would not be satisfied. Without loss of generality, by loop erasure, this cycle can be chosen to be a simple cycle. We decrease the current uniformly along every edge of this cycle until we get an edge with 0 current. The net effect of this procedure is to remove an edge and its associated current. Iterating the procedure $\tilde{\ul{j}}$ shows that $\tilde{\mathcal C}$ can be represented as an element of $\mathcal{C}$. Since $\ul{j} \in \mathcal{C}$ was arbitrary, this gives an isomorphism and hence Theorem \ref{thm:K}.
\end{proof}



\subsection{Lattice point geometry}

Recall that in the main text, the allowed currents are necessarily integer-valued (Definition \ref{defn:Feynman}). We note that Theorem \ref{thm:K} implies a result about the number of integer-valued solutions as well.

\be{cor}[Number of integer-valued solutions to Kirchhoff's law]
\label{cor:K}
Let $L$ be an exploration on $l$ vertices and $k$ edges. Then, there is a constant $C_{k}>0$ such that 
$$
|\setof{\ul{j}=(j_1,\ldots, j_k)\subset \{1,\ldots,M\}}{\ul{j}\textnormal{ is $L$-admissible}}|\leq C_{k}M^{k+1-l}.
$$
\e{cor}

The corollary is implied by the following statement by taking $V=\curly{C}$. (The corollary is not used in the main text, but its refinement, Lemma \ref{lem:latticepoints2} below, will be used.)

We recall that an affine space is a shifted linear subspace.

\be{lm}\label{lem:latticepoints}
There exists a universal constant $C>0$ such that for every affine space $V\subset \R^k$ of dimension $1\leq d\leq k$, it holds that
$$
|V\cap \{1,\ldots, N\}^k|\leq C_d d^{k/2} N^d.
$$
\end{lm}
\begin{proof}
We translate the counting of lattice points to a statement about volumes by using balls centered at lattice points. For each point $p\in \{1,\ldots, N\}^k$, we define the ball $B^k(p)\subset \R^k$ as the $k$-dimensional ball of radius $1/2$ centered at $p$, and we define $B^d(p)\subset V$ as the $d$-dimensional ball of radius $1/2$ centered at $p$ in the subspace $L$. Note that these balls are pairwise disjoint, i.e., $B^k(p)\cap B^k(p')=\emptyset$ if $p\neq p'$ (and similarly for $B^d$). Moreover, for any $p\in V$, we have
$$
B^k(p)\cap V=B^d(p).
$$
Thus, we can bound the number of points in $V\cap \{1,\ldots, N\}^k$ by the ratio of the volume of $V\cap [1,N]^d$ to the volume of a $B^d(p)$. This ratio is bounded by $C_d d^{k/2} N^d$ for a constant $C_d>0$, which proves the lemma.
\e{proof}
We now refine this argument. The following lemma is used for counting the number of admissible assignments of external currents in the proof of Lemma \ref{lem:ESBtoPhi}.

\begin{lm}\label{lem:latticepoints2}
Let $1\leq m\leq m'$ and fix a basis $\{b_1,\ldots,b_m,b_{m+1},\ldots,b_{m'}\}$ of $\R^{m'}$ such that $b_{m+1},\ldots,b_{m'}$ have coordinates in the set $\{0,\pm 1\}$. Fix a vector $\tau\in \R^{m'}$. Then the set 
\beq\label{eq:latticepointform}
\setof{(a_1,\ldots,a_m)\in \R^m}{\exists\, a_{m+1},\ldots,a_{m'}\in \R:\, \tau+\sum_{i=1}^{m'} a_i b_i\in \{1,\ldots,N\}^{m'}}
\eeq
has cardinality bounded by $C_{m'} N^m$.
\end{lm}

\begin{proof}
Without loss of generality, we assume that $N$ is even and that we require the weaker condition $\sum_{i=1}^{m'} a_i b_i\in \{0,\ldots,N\}^{m'}$. First, we choose $\tau=(N/2,\ldots,N/2)$. That is, we consider the set
$$
\begin{aligned}
S_{N/2}=&\setof{(a_1,\ldots,a_m)\in \R^m}{\exists\, a_{m+1},\ldots,a_{m'}\in \R:\, \tau+\sum_{i=1}^{m'} a_i b_i\in \{0,\ldots,N\}^{m'}}\\
=&\setof{(a_1,\ldots,a_m)\in \R^m}{\exists\, a_{m+1},\ldots,a_{m'}\in \R:\, \sum_{i=1}^{m'} a_i b_i\in \{-N/2,\ldots,N/2\}^{m'}}.
\end{aligned}
$$
Now fix a vector $(a_1,\ldots,a_m)\in S_{N/2}$, with an associated collection $a_{m+1},\ldots,a_{m'}$. Since $b_{m+1},\ldots,b_{m'}$ have coordinates in the set $\{0,\pm 1\}$, the triangle inequality implies
$$
\sum_{r=1}^{m'} a_i b_i+\sum_{i=m+1}^{m'} \al_i b_i \in \{-N,\ldots,N\}^{m'}
$$
for all $\al_{m+1},\ldots,\al_{m'}\in \{-N/(2{m'}),\ldots,N/(2{m'})\}$. This shows that to every vector $(a_1,\ldots,a_m)\in S_{N/2}$, we can associate $(N/{m'})^{{m'}-m}$ points in $\{-N,\ldots,N\}^{m'}$. Note that these associated points are different for every vector $(a_1,\ldots,a_m)\in S_{N/2}$ since $\{b_1,\ldots,b_{m'}\}$ form a basis of $\R^{m'}$. Since the cardinality of $\{-N,\ldots,N\}^{m'}$ is $(2N+1)^{m'}$, the pigeonhole principle implies that the cardinality of $S_{N/2}$ is bounded by
$$
(2N+1)^{m'}\l(\frac{{m'}}{N}\r)^{{m'}-m}\leq C_{m'} N^m.
$$
Finally, we observe that the argument generalizes to an arbitrary shift vector $\tau$ by shifting the box $\{-N,\ldots,N\}^{m'}$ by $\tau$. This proves the lemma.
\end{proof}

\section{Bypass lemmas}
In this section, as in Section \ref{sect:graphtheory}, we work with ordinary graphs $G=(V,E)$ with each vertex having even degree and undirected edges.

\begin{lm}[Bypass lemma]\label{lem:bypass}
Let $G= (V,E)$ be an undirected graph. Fix a simple cycle $C$ and two edges $e,e' \in C$. Suppose there exists a simple cycle $C'$ which contains $e'$ but not $e$. Then, there is a simple cycle $\tilde{C}$ which, conversely, contains $e$ but not $e'$. The edges of the cycle $\tilde{C}$ will be contained in $C\cup C'$.
\end{lm}
\begin{figure}
    \centering
    \includegraphics[scale=0.3]{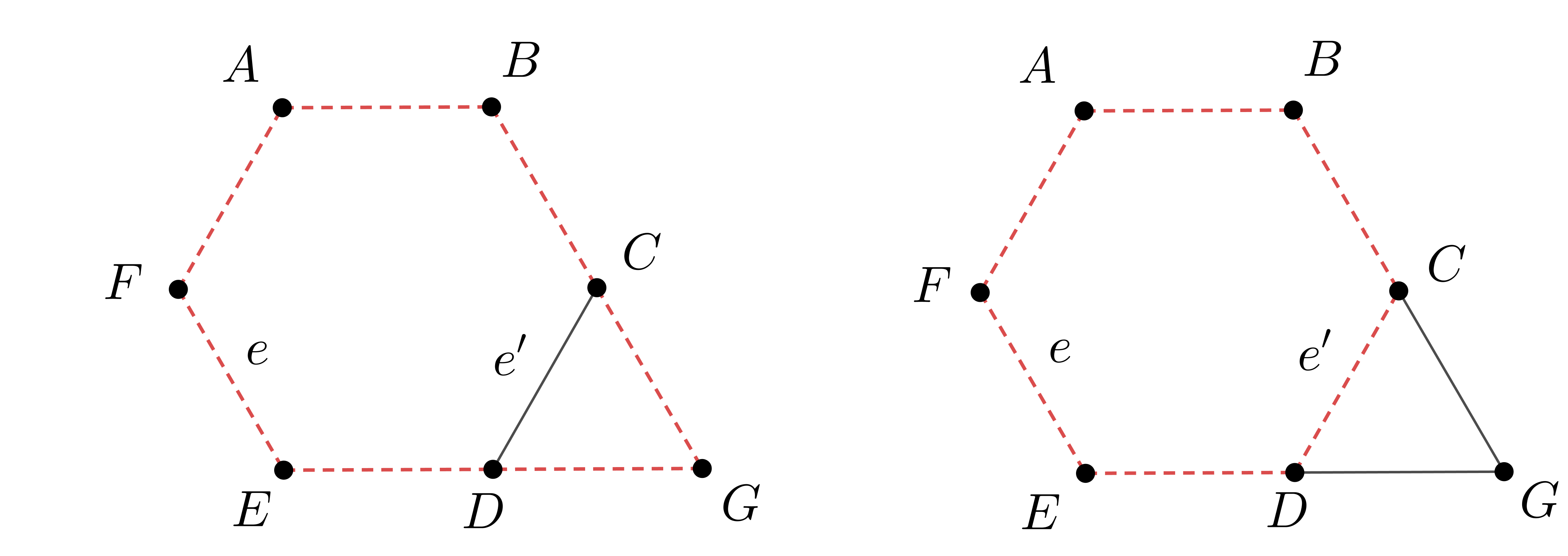}
    \caption{This explains the statement of the bypass lemma \ref{lem:bypass}. In passing from the left image to the right, the edge $e'=C\to D$ was bypassed via the path $C\to G\to D$. Since $C\to G\to D$ does not contain the edge $e$, one has that $e$ is still contained in the red cycle after bypassing.}
    \label{fig:Bypass}
\end{figure}

The name of the lemma derives from the image that the cycle $\tilde C$ ``bypasses'' the edge $e'$ by taking a detour along the cycle $C'$. See Figure \ref{fig:Bypass} for an example.

\begin{proof}
The cycle $C$ can be written as 
$$
C= e\to p_{e_1\rightarrow e_1'}\to e'\to p_{e_2' \rightarrow e_2},
$$
where $e_1,e_2$ are the neighbors of $e$ in $C$ and $e_1',e_2'$ are the neighbors of $e'$ in $C$, and $p_{e\rightarrow f}$ is the simple path in $C$ whose first edge is $e$ and whose last edge is $f$.   Similarly, the cycle $C_{e'}$ can be written as
$$
C'=e'\to p'_{f_1'\to f_2'},
$$
$f_1',f_2'$ are the neighbors of $e'$ in $C'$ and $p'_{f_1'\to f_2'}$ is the simple path in $C'$ whose first edge is $f_1'$ and whose last edge is $f_2'$.

Now consider the composition $e\to p_{e_1\rightarrow e_1'}\to p'_{f_1'\to f_2'} \to p_{e_2' \rightarrow e_2}$. This is a possibly non-simple cycle that uses the edge $e$ only once and does not use the edge $e'$. Finally, we apply loop erasure to reduce this to a simple cycle, where loop erasure was defined in the proof of Proposition \ref{prop:typeIIodd}. This proves Lemma \ref{lem:bypass}.
\end{proof}

\begin{lm}[Even bypass lemma] \label{lem:EvBypass}
Consider a graph $G=(V,E)$. Take two distinct edges $e_1$ and $e_2$ and assume that there is a cycle $C$, which will necessarily not be simple, that uses the edge $e_1$ exactly once and the edge $e_2$ an even number of times. Then there exists a cycle $\tilde{C}\subset C$ that uses the edge $e_1$ only once and does not use the edge $e_2$. 
\end{lm}
\begin{proof}
We will induct on the number of times $2k$ that $e_2$ appears in the cycle $C$. The case $k=0$ is trivial. Assume that the claim holds for $k \le n-1$, we will now proceed to show the claim for $k=n$.

First assign an orientation to the cycle $C$. We will write
\begin{equation}
    C= \hat{e}_0 \rightarrow \hat{e}_1 \rightarrow \hat{e}_2\rightarrow \ldots\rightarrow \hat{e}_m \rightarrow \hat{e}_0
\end{equation}
where we have set $\hat{e}_0= e_1$. In this ordering, let $\hat{e}_i$ be the first appearance of the edge $e_2$ and let $\hat{e}_f$ be the final appearance of $e_2$. Let the two endpoints of the edge $e_2$ be $v$ and $w$. 

First consider the case that both $\hat{e}_i$ and $\hat{e}_f$ are oriented in in the opposite direction; namely, $\hat{e}_i= v \rightarrow w$ and $\hat{e}_f= w \rightarrow v$ or vice versa. We can then define 
\begin{equation}
    \tilde{C}:= \hat{e}_0 \rightarrow\ldots\rightarrow \hat{e}_{i-1} \rightarrow \hat{e}_{f+1}\rightarrow \ldots\rightarrow \hat{e}_m \rightarrow \hat{e}_0
\end{equation}
and we are done. Indeed, we are able to skip from the left endpoint of $\hat{e}_i$ directly to the right endpoint of $\hat{e}_f$. The resulting graph will have no appearance of the edge $e_2$.

We now only need to consider the case that $\hat{e}_i$ and $\hat{e}_f$ are oriented in the opposite direction. We will define
\begin{equation}
    \hat{C}:= \hat{e}_0 \rightarrow \ldots\rightarrow \hat{e}_{i-1} \rightarrow \hat{e}_{f-1} \rightarrow \hat{e}_{f-2}\rightarrow\ldots\rightarrow \hat{e}_{i+1} \rightarrow \hat{e}_{f+1}\rightarrow\ldots \to \hat{e}_m \rightarrow \hat{e}_0
\end{equation}
In words, $\hat{C}$ is constructed by removing the two edges $\hat{e}_i$ and $\hat{e}_f$ and connecting $\hat{e}_{i-1}$ to $\hat{e}_{f+1}$ by using the reverse of the path between $\hat{e}_{i+1}$ to $\hat{e}_{f-1}$. Notice that the number of appearances of $e_2$ in $\hat{C}$ is $2(n-1)$. Hence, we can apply the induction hypothesis to $\hat{C}$ and this proves the induction step.
\end{proof}
\end{appendix}

\section*{Acknowledgments}
The authors are grateful to Noam D.~Elkies for useful advice concerning the exponential sum estimates in Section \ref{sect:ntheory}. The work of H.-T.~Y.\ is partially supported by NSF Grant DMS-1606305 and a Simons Investigator award.


\end{document}